\documentclass{lmcs}
\pdfoutput=1

\usepackage{lastpage}
\renewcommand{\dOi}{14(4:24)2018}
\lmcsheading{}{1--\pageref{LastPage}}{}{}%
{Mar.~09,~2017}{Dec.~11,~2018}{}

\begin{document}

\title{Subsumption Algorithms for Three-Valued Geometric Resolution}
\author{Hans de Nivelle}
\address{School of Science and Technology, Nazarbayev University,
         53 Qabanbay Batyr, Astana 010000, Kazakhstan}

\begin{abstract}
   In our implementation of geometric resolution, the
   most costly operation is subsumption testing (or matching):
   One has to decide for a three-valued, geometric formula,
   if this formula is false in a given interpretation.
   The formula contains only atoms with variables, equality, and existential
   quantifiers. The interpretation contains only atoms with constants.
   Because the atoms have no term structure, 
   matching for geometric resolution is hard.  
   We translate the matching problem into a generalized constraint 
   satisfaction problem, and discuss several approaches for solving 
   it efficiently,
   one direct algorithm and two translations to propositional SAT. 
   After that, we study filtering techniques based on local 
   consistency checking.
   Such filtering techniques can a priori refute a large percentage
   of generalized constraint satisfaction problems.
   Finally, we adapt the matching algorithms in
   such a way that they find solutions that use a minimal subset of
   the interpretation. The adaptation can be combined with
   every matching algorithm. 
   The techniques presented in this paper may have applications
   in constraint solving independent of geometric resolution.
\end{abstract}

\maketitle
\section{Introduction}
\label{Sect_introduction}

Main topic of this paper is \emph{the generalized matching problem},
for example how to match $ p(X,Y), \ q(Y,Z) $ into
$ p(0,1), \ p(0,2), \ q(1,3), \ q(2,4), \ r(0,3) $ without 
matching $ r(X,Z). $ This problem arose in the implemention
of geometric resolution. 
Geometric logic as a theorem proving strategy was introduced
in \cite{BezemCoquand2005}. (The authors use the name 
\emph{coherent logic}.)
Bezem and Coquand were motivated mostly by the desire to
obtain a theorem proving strategy with a 
simple normal form transformation, which makes that many natural
problems need no transformation at all, others have a much simpler
transformation, and which makes that in all cases
Skolemization can be avoided. This results in more readable proofs,
and proofs that can be backtranslated more easily. 

Our motivation for using geometric resolution is different,
more engineering-oriented: We hope that three-valued, geometric 
resolution can be made sufficiently efficient, so that it can
be used as a generic reasoning core, into which different kinds 
of two- or three-valued decision
problems (e.g. problems representing type correctness, two-valued decision 
problems, or simply typed classical problems) 
can be translated. Because we want the geometric reasoning core 
to be generic, we are willing to accept transformations that do not 
preserve much
of the structure of the original formula. Subformulas are freely renamed,
and functional expressions are flattened and replaced by relations. 
For details of the calculus, its motivation, and related work,
we refer to \cite{deNivelle2014b}. 

In the current paper we give only a short introduction, 
which is aimed at explaining how matching is used in geometric 
resolution, and how matching instances in geometric resolution
are translated into generalized constraint satisfaction problems.  
If one is interested only in the methods for constraint satisfaction, 
one can ignore the technical part of this section and continue
reading at the overview at the end of this section.

We continue this section by giving 
a definition of three-valued, geometric formulas. 
The definition that we give here is slightly too general, but easier 
to understand than the correct definition in \cite{deNivelle2014b}, 
which contains some additional, technical restrictions 
which are not relevant for matching.

\begin{defi}
   \label{Def_geometric_atom}
   A \emph{geometric literal} has one of the following four forms: 
   \begin{enumerate}
   \item
      A simple atom of form $ p_{\lambda}( x_1, \ldots, x_n ), $
      where $ x_1, \ldots, x_n $ are variables (with repetitions allowed) 
      and $ \lambda \in \{ {\bf f}, \ {\bf e}, \ {\bf t} \}. $
      (denoting \emph{false}, \emph{error} and \emph{true}.)
   \item
      An equality atom of form $ x_1 \approx x_2, $ with $ x_1,x_2 $ 
      distinct variables.
   \item
      A domain atom $ \#_{\bf f} \, x, $ with $x$ a variable. 
   \item
      An \emph{existential atom} of form
      $ \exists y \ p_{\lambda}( x_1, \ldots, x_n, y ) $ with
      $ \lambda \in \{ {\bf f}, \ {\bf e}, \ {\bf t} \}, $ and 
      such that $ y $ occurs at least once in the atom, not necessarily
      on the last place. 
   \end{enumerate}

   \noindent
   A \emph{geometric formula} has form 
   $ A_1, \ldots, A_p \ | \ B_1, \ldots, B_q, $ where 
   the $ A_i $ are simple or domain atoms, and the
   $ B_j $ are atoms of arbitrary type.

   We require that geometric formulas are \emph{range restricted},
   which means that every variable that occurs free in a $ B_j $ 
   must occur in an $ A_i $ as well. 

   The intuitive meaning of $ A_1, \ldots, A_p \ | \ B_1, \ldots, B_q $ is
   $ \forall \overline{x} \ A_1 \vee \cdots \vee A_p \vee 
                            B_1 \vee \cdots \vee B_q, $ 
   where $ \overline{x} $ are all the free variables.
   The vertical bar $ (|) $ has no logical meaning. Its only purpose is 
   to separate the two types of atoms. 
\end{defi}

\noindent
A geometric formula that is not range restricted, can always be made range 
restricted by inserting suitable $ \#_{\bf f} $ atoms into the left hand side. 
This is the only purpose of the $ \# $-predicate.
Interpretations contain predicates of form $ \#_{\bf t} \ c, $ 
for every domain element $ c. $ 
Atoms in geometric formulas are variable-only, and are labeled with 
truth-values, as in \cite{MurrayRosenthal1993}. It is shown in 
\cite{deNivelle2014a}
and \cite{deNivelle2014b} that formulas in classical logic with
partial functions (\cite{deNivelle2011a}) can be translated into sets 
of geometric formulas.

\begin{defi}
   \label{Def_interpretation}
   We define an \emph{interpretation} $ I $ as a finite set of atoms
   of forms $ \#_{\bf t} \, c $ with $ c $ a constant, or 
   form $ p_{\lambda}( x_1, \ldots, x_n ), $ where $ x_1, \ldots, x_n $
   are constants (repetitions allowed).
   Interpretations must be \emph{range restricted} as well.
   This means that every constant $ x $ occurring in the
   interpretation must occur in an atom of form $ \#_{\bf t} \, x. $ 
\end{defi}

\noindent
Matching searches for false formulas. These are formulas whose
premises $ A_1, \ldots, A_p $ clash with $ I, $ while none of the
$ B_j $ is true in $ I. $ 
 
\begin{defi}
   \label{Def_conflict_truth} 
   Let $ I $ be an interpretation. Let $ A $ be a geometric literal.  
   Let $ \Theta $ be a substitution that assigns constants to variables,
   and that is defined on the variables in $ A. $ 
   We say that $ A \Theta $ \emph{conflicts}(or \emph{is in conflict with}) 
   $ I $ if {\bf (1)} 
   $ A $ has form $ p_{\lambda}( x_1, \ldots, x_n ), $ and there is an
   atom of form $ p_{\mu}( x_1 \Theta, \ldots, x_n \Theta ) \in I $ 
   with $ \lambda \not = \mu, $ \ 
   {\bf (2)} 
   $ A $ has form $ x_1 \approx x_2 $ and $ x_1 \Theta \not = x_2 \Theta, $ or
   {\bf (3)} 
   $ A $ has form 
   $ \#_{\bf f} \, x $ and $ ( \#_ {\bf t} \, x \Theta ) \in I. $
 
   \noindent 
   We say that $ A \Theta $ \emph{is true} in $ I $ if \\
   {\bf (1)} $ A $ has form $ p_{\lambda}( x_1, \ldots, x_n ) $ and
   $ p_{\lambda}( x_1 \Theta, \ldots, x_n \Theta ) \in I, $  \
   {\bf (2)} $ A $ has form $ x_1 \approx x_2 $ and 
             $ x_1 \Theta = x_2 \Theta, $ \ 
   {\bf (3)} $ A $ has form $ \#_{\bf t} \, x  $ and 
             $ ( \#_{\bf t} \, x \Theta ) \in I, $ 
   or \ {\bf (4)} 
             $ A $ has form $ \exists y \ B_{\lambda}( x_1, \ldots, x_n, y ) $ 
   and there exists a constant $ c, $ s.t.
   $ B_{\lambda}( x_1 \Theta, \ldots, x_n \Theta, c ) \in I. $ 
\end{defi}
In the definitions of truth and conflict, $ \# $ is treated as a usual 
predicate. 

\begin{defi}
   Let $ I $ be an interpretation. Let $ B $ be a geometric atom.
   Let $ \Theta $ be a substitution that instantiates all free variables of 
   $ B, $ 
   and for which $ B \Theta $ is not true in $ I. $ 
   We define the \emph{extension set} $ E( B, \Theta ) $ as follows:
   \begin{itemize}
   \item
      If $ B $ has form $ p_{\lambda}( x_1, \ldots, x_n) $ or
      $ \#_{\bf t} \ x, $ then
      $ E( B, \Theta ) = \{ B \Theta \}. $ 
   \item
      If $ B $ has form $ x_1 \approx x_2, $ then 
      $ E( B, \Theta ) = \emptyset. $ 
   \item
      If $ B $ has form $ \exists y \ p_{\lambda}( x_1, \ldots, x_n, y ), $ then
      \[ E( B, \Theta ) = 
         \{ \ p_{\lambda}(x_1 \Theta, \ldots,x_n \Theta ,c) \ \} \ | \ 
                       c \in I \ \} \cup 
            \{ \ p_{\lambda}(x_1 \Theta,\ldots,x_n \Theta, \hat{c}) \ \} \ \}. \]
   By $ c \in I $ we mean: $ c $ is a constant occurring in 
   an atom of $ I. $ 
   We assume that $ \hat{c} $ is a fresh constant
   for which $ \hat{c} \not \in I. $ 
   \end{itemize}
\end{defi}
\noindent Intuitively, if for a geometric formula 
$ \phi = \ A_1, \ldots, A_p \ | \ B_1, \ldots, B_q $ and a substitution 
$ \Theta, $ the $ A_i \Theta $ are in conflict with $ I, $ while none of the 
$ B_j \Theta $ is true in $ I, $ then $ \phi \Theta $ is false in $ I. $ If 
there exist a $ B_j $ and an atom $ C \in E( B_j, \Theta ) $ that is not
in conflict with $ I, $ then $ \phi \Theta $ can be made true by
adding $ C. $ If no such $ C $ exists, a conflict was found. 
If more than one $ C $ exists, the search algorithm has to
backtrack through all possibilities. 
The search algorithm tries to extend an initial interpretation $ I $ into
an interpretation $ I' \supset I $ that makes all formulas true.
At each stage of the search, it looks 
for a formula and a substitution that make the formula false. 
If no formula and substitution can
be found, the current interpretation is a model. Otherwise, search 
continues either by extending $ I, $ or by backtracking.  
Details of the procedure are 
described in \cite{deNivelleMeng2006a} for the two-valued case, 
and in \cite{deNivelle2014b} for the three-valued case. 
Experiments with the current three-valued version (available from
\cite{CASC2016}), and the previous
two-valued version (\cite{deNivelleMeng2007f}) show that 
the search for false formulas consumes nearly all of the
resources of the prover. 

\begin{defi} 
   \label{Def_matching}
   An instance of \emph{the matching problem} consists of 
   an interpretation $ I $ and a geometric formula 
   $ A_1, \ldots, A_p\ | \ B_1, \ldots, B_q. $
  
   Determine if there exists a substitution $ \Theta $ that brings
   all $ A_i $ in conflict with $ I, $ and makes none of the
   $ B_j $ true in $ \Theta. $ 
   If yes, then return such substitution. 
\end{defi}

\begin{exas}
   \label{Ex_matchings}
   Consider an interpretation $ I $ consisting of atoms
   \[ P_{\bf t}( x_0, x_0 ), \ P_{\bf e}( x_0, x_1 ), \ 
      P_{\bf t}( x_1, x_1 ), \ P_{\bf e}( x_1, x_2 ), \ 
      Q_{\bf t}( x_2, x_0 ). \]

   \noindent
   The formula 
   $ \phi_1 = \ P_{\bf f}(X,Y), \ P_{\bf f}(Y,Z) \ | \ 
                Q_{\bf t}(Z,X) $ can be matched in five ways:
   \[ \begin{array}{l}
         \Theta_1 = \{ \ X := x_0, \ Y := x_0, \ Z := x_0 \ \} \\
         \Theta_2 = \{ \ X := x_0, \ Y := x_0, \ Z := x_1 \ \} \\
         \Theta_3 = \{ \ X := x_0, \ Y := x_1, \ Z := x_1 \ \} \\
         \Theta_4 = \{ \ X := x_1, \ Y := x_1, \ Z := x_1 \ \} \\
         \Theta_5 = \{ \ X := x_1, \ Y := x_1, \ Z := x_2 \ \} \\
      \end{array}  
   \]
   The substitution 
   $ \Theta_6 = \{ \ X := x_0, \ Y := x_1, \ Z := x_2 \ \} $ would make
   the conclusion $ Q_{\bf t}(Z,X) $ true. 
   Next consider the formula 
   $ \phi_2 = \ P_{\bf f}( X,Y ), \ P_{\bf t}(Y,Z) \ | \ 
                X \approx Y. $ \\
   The substitution $ \Theta = \{ \ X := x_0, \ Y := x_1, \ Z := x_2 \ \} $ 
   is the only matching of $ \phi_2 $ into $ I. $ 
   Finally, the formula
   $ \phi_3 = \ P_{\bf t}(X,Y) \ | \ \exists Z \ Q_{\bf t}(Y,Z) $ 
   can be matched with $ \Theta = \{ \ X := x_0, \ Y := x_1 \ \}, $ 
   and in no other way.
\end{exas}

\noindent
The first formula $ \phi_1 $ in example~\ref{Ex_matchings} has five matchings. 
In case there exists more than one matching, it matters for 
the geometric prover which matching is returned.
This is because the prover analyses which ground atoms in the interpretation
$ I $ contributed to the matching, and will consider only those 
in backtracking.
In general, the set of conflicting atoms in $ I $ should be 
as small as possible,
and should depend on as few as possible decisions. (Decisions in the 
sense of propositional reasoning, see \cite{Handbook_sat:CDCL}.) 
The simplest solution for finding
the best matching would be
to enumerate all matchings, and use some preference relation
$ \preceq $ to keep the best one.
Unfortunately, this approach is not practical because the number 
of matchings
can be extremely high. 
We will address this problem in Section~\ref{Sect_optimal}. 

Even if one is interested in the decision problem only,
matching is still intractable because the decision problem is 
already NP-complete.  
(See problem {\bf LO18} in \cite{GareyJohnson79}.)
In this paper, we introduce several algorithms for 
efficiently solving the matching problem. The algorithms
evolved out of predecessors that have been implemented before
in the two-valued version of {\bf Geo} 
(\cite{deNivelleMeng2007f}), and in the three-valued
version of {\bf Geo} that took part in CASC~J8
(see \cite{CASC2016}). The matching algorithm of the three-valued
version is discussed in detail in \cite{deNivelle2016a}.
Unfortunately, after comparison with other methods,
in particular the algorithms in the current paper,
and translation to SAT, the approach of \cite{deNivelle2016a}
turned out not competitive, and we have abandoned it.
The algorithm in this paper, and translation to SAT
are on average
500-1000 times faster than the algorithm of \cite{deNivelle2016a}.

The paper is organized as follows: 
In Section~\ref{Sect_GCSP}, we translate the matching
problem into a structure called \emph{generalized constraint satisfaction
problem} (GCSP). The generalization consists of the fact that it contains
additional constraints, that a solution must not make true. These constraints 
correspond to the conclusions of the geometric formula 
that one is trying to match.

After that, we present in Section~\ref{Sect_unary_algo} a backtracking 
algorithm for solving GCSP, which is based on backtracking combined
with a form of propagation. It relies on a data structure that we call 
\emph{refinement stack}. 
Refinements stacks were introduced
in \cite{deNivelle2016a}. The matching algorithm of \cite{deNivelle2016a}
turned out non-competitive, but its data structure is still useful. 
In Section~\ref{Sect_conflict_learning} we add conflict
learning to our matching algorithm.
In Section~\ref{Sect_IJCAR2016}, we briefly discuss the
algorithm of \cite{deNivelle2016a}. In Section~\ref{Sect_SAT_trans},
we give two translations from GCSP to SAT. The translations are 
straightforward, and efficiently solved by MiniSat (\cite{Minisat2004}). 
In order to make it possible to run our matching algorithm
independent of geometric logic, possibly opening the way
for other applications, we define an input format for
matching problems in Section~\ref{Sect_input_format}.
The format is derived from the DIMACS format for SAT.
We released the sources in \cite{deNivelle2018a}.
Section~\ref{Sect_experiments} contains experimental results.
The main conclusions are that the algorithm of \cite{deNivelle2016a}
is not competitive, and that our own algorithm is comparable
to translation to SAT combined with MiniSat. 
In Section~\ref{Sect_optimal}, we explain how every algorithm
that is able to find some solution, can be transformed into an
algorithm that finds an optimal solution. This transformation
is essential for the application in geometric resolution.
In Section~\ref{Sect_local_consistency_checking},
we present a priori filtering techniques, that are able to
reject a large percentage of matching instances a priori.

\section{Translation into Generalized Constraint Satisfaction Problem}
\label{Sect_GCSP}

We introduce the generalized constraint satisfaction problem,
and show how instances of the matching problem can be translated.
It is `generalized' because there are additional, negative
constraints (called \emph{blockings}), which a solution is not
allowed to satisfy. The blockings originate from translations of
the $ B_1, \ldots, B_q. $ 

\begin{defi}
   \label{Def_substlet}
   A \emph{substlet} $ s $ is a (small) substitution.
   We usually write $ s $ in the form $ \overline{v} / \overline{x}, $
   where $ \overline{v} $ is a sequence of variables without repetitions,
   and $ \overline{x} $ is a sequence of constants of same length as
   $ \overline{v}. $ 

   We say that two substlets $ \overline{v}_1/\overline{x}_1 $ and
           $ \overline{v}_2/\overline{x}_2 $ are \emph{in conflict}
   if there exist $ i,j $ s.t. $ v_{1,i} = v_{2,j} $ and 
   $ x_{1,i} \not = x_{2,j}. $

   If $ \overline{v}_1/\overline{x}_1, \ldots, \overline{v}_n/\overline{x}_n $
   is a sequence of substlets not containing a conflicting pair, then
   one can merge them into a substitution as follows: 
   $ \bigcup \{ \overline{v}_1/\overline{x}_1, \ldots, 
                \overline{v}_n/\overline{x}_n \} = 
   \{ 
      v_{i,j} := x_{i,j} \ | \ 1 \leq i \leq n, \ 
                               1 \leq j \leq \| \overline{v}_i \| \}. $ 

   If $ \Theta $ is a substitution and $ s = \overline{v} / \overline{x} $ 
   is a substlet, we say that $ \Theta $ \mbox{makes} $ s $ \emph{true} if 
   every $ v_i := x_i $ is present in $ \Theta. $ 

   We say that $ \Theta $ and $ s $ \mbox{are in conflict} if 
   there is a $ v_i/x_i $ with $ 1 \leq i \leq \| v \|, $ s.t. 
   $ v_i \Theta $ is defined and distinct from $ x_i. $ 

   A \emph{clause} $ c $ is a finite set of substlets
   with the same domain. 
   We say that a substitution $ \Theta $ \emph{makes} $ c $ \emph{true} 
   (notation $ \Theta \models c $)
   if $ \Theta $ makes a substlet 
   $ ( \overline{v} / \overline{x} ) \in c $ true. 
   We say that $ \Theta $ \emph{makes} $ c $ \emph{false}
   (notation $ \Theta \models \neg c $) if 
   every substlet $ ( \overline{v} / \overline{x} ) \in c $ is in conflict
   with $ \Theta. $ 
   In the remaining case, we call $ c $ \emph{undecided by} $ \Theta. $ 
\end{defi}

\begin{defi} 
   \label{Def_gcsp} 
   A \emph{generalized constraint satisfaction problem} (GCSP) 
   is a pair of form $ ( \Sigma^{+}, \Sigma^{-} ) $ in which
   $ \Sigma^{+} $ is a finite set of clauses,
   and $ \Sigma^{-} $ is a finite set of substlets.

   A substitution $ \Theta $ is a \emph{solution} of 
   $ ( \Sigma^{+}, \Sigma^{-} ), $ if 
   every clause in $ \Sigma^{+} $ is true in $ \Theta, $ and 
   there is no $ \sigma \in \Sigma^{-}, $ s.t. $ \Theta $ makes
   $ \sigma $ true.
\end{defi}

\begin{defi}
   Let $ ( \Sigma^{+}, \Sigma^{-} ) $ a GCSP.
   We call $ ( \Sigma^{+}, \Sigma^{-} ) $ 
   \emph{range restricted} if for every variable $ v $ that occurs in
   a substlet $ \sigma \in \Sigma^{-}, $ there exists a clause
   $ c \in \Sigma^{+} $ s.t. every substlet $ s \in c $ has
   $ v $ in its domain.
\end{defi}

\noindent
We now explain how a matching instance is translated into 
a generalized constraint satisfaction problem. 

\begin{defi}
   \label{Def_trans_gcsp}
   Assume that $ I $ and $ \phi = \ A_1, \ldots, A_p \ | \ B_1, \ldots, B_q  $
   together form an instance of the matching problem.
   The \emph{translation} $ ( \Sigma^{+}, \Sigma^{-} ) $ 
   of $ ( I, \phi ) $ 
   \emph{into} GCSP is obtained as follows:
   \begin{itemize}
   \item
      For every $ A_i, $ let $ \overline{v}_i $ denote the 
      variables of $ A_i. $ 
      Then $ \Sigma^{+} $ contains the clause
      \[ \{ \ \overline{v}_i / \overline{v}_i \Theta \ | \ 
            A_i \Theta \mbox{ is in conflict with } I \ \}. \]
   \item
      For every $ B_j, $ let $ \overline{w}_j $ denote the variables
      of $ B_j. $ 
      For every $ \Theta $ that makes $ B_j \Theta $ 
      true in $ I, $ \ \ $ \Sigma^{-} $ contains the substlet
      $ \overline{w}_j / ( \overline{w}_j \Theta ). $ 
   \end{itemize}
\end{defi}

\begin{thm}
   A matching instance $ ( I, \phi ) $ has a matching iff its corresponding 
   GCSP has a solution.
\end{thm}

\noindent
In theory, the set of blockings $ \Sigma^{-} $ can be removed,
because a blocking $ \sigma $ can always be replaced
by a clause as follows: Let $ \sigma $ be a blocking, 
let $ \overline{v} $ be 
its variables. Define $ \sigma_1 = \sigma, $ and let
$ \sigma_2, \ldots, \sigma_n \in \Sigma^{-} $ be the blockings
whose domain is also $ \overline{v}. $ 
One can replace $ \sigma_1, \ldots, \sigma_n $ by 
the clause $ \{ \ \overline{v}/\overline{c} \ | \ 
   \overline{v} / \overline{c} \mbox{ conflicts all }
      \sigma_i \ ( 1 \leq i \leq n ) \ \}. $ 

We prefer to keep $ \Sigma^{-}, $ because in the worst case, 
the resulting clause has size $m^{ \| \overline{v} \| }, $ 
where $ m $ is the size of the domain.
For example, if $ \sigma_1, \ldots, \sigma_n $ result from an
equality $ X \approx Y, $ then $ \sigma_i $ has form
$ ( X,Y ) / ( x_i, x_i ). $ 
The resulting clause $ c = \{ (X,Y)/(x_i, x_j) \ | \ i \not = j \} $ 
has size $ n(n-1) \approx n^2. $ 

Clauses resulting from a matching problem have the following
trivial, but essential property:
\begin{lem}
   \label{Lemma_conflict_in_clause}
   Let $ ( \Sigma^{+}, \Sigma^{-} ) $ be obtained 
   by the translation in Definition~\ref{Def_trans_gcsp}. 
   Let $ s_1,s_2 \in c \in \Sigma^{+}. $ 
   Then either $ s_1 = s_2, $ or $ s_1 $ and $ s_2 $ are 
   in conflict with each other. 
\end{lem}
Lemma~\ref{Lemma_conflict_in_clause} holds because $ s_1 $ and $ s_2 $ 
have the same domain. 

\begin{exas}
   \label{Ex_translations}
   In example~\ref{Ex_matchings}, the matching
   problem $ ( I, \phi_1 ) $ can be translated into 
   the GCSP below. The clauses are above the horizontal
   line, and the blockings are below it. Because substlets
   in the same clause always have the same variables, we write
   the variables of a clause only once. 
   \[ \begin{array}{l}
         (X,Y) \ / \ ( x_0,x_0 ) \ | \ ( x_0,x_1 ) \ | \ 
                     ( x_1,x_1 ) \ | \ ( x_1,x_2 ) \\
         (Y,Z) \ / \ ( x_0,x_0 ) \ | \ ( x_0,x_1 ) \ | \ 
                     ( x_1,x_1 ) \ | \ ( x_1,x_2 ) \\
         \hline
         (X,Z) \ / \ ( x_0,x_2 ) \\
      \end{array}
   \]
   Translating $ (I,\phi_2) $ results in: 
   \[ \begin{array}{l}
         (X,Y) \ / \ ( x_0,x_0 ) \ | \ ( x_0,x_1 ) \ | \
                     ( x_1,x_1 ) \ | \ ( x_1,x_2 ) \\
         (Y,Z) \ / \ ( x_0,x_1 ) \ | \ ( x_1,x_2 ) \\
         \hline
         (X,Y) \ / \ ( x_0, x_0 ) \\
         (X,Y) \ / \ ( x_1, x_1 ) \\
         (X,Y) \ / \ ( x_2, x_2 ) \\
      \end{array}
   \]
   Translation of $ (I,\phi_3) $ results in: 
   \[ \begin{array}{l}
         (X,Y) \ / \ ( x_0,x_1 ) \ | \ ( x_1,x_2 ) \\
         \hline
         (Y) \ / \ ( x_2 ) \\
      \end{array}
   \]
\end{exas}

\noindent
Before one runs any algorithms on a GCSP, it is useful to do
some simplifications. If the GCSP contains a propositional clause
(a clause whose domain contains no variables), this clause 
either has form $ ( \ ) \ / \ $ (no assignments), or
$ ( \ ) \ / \ ( \ ) $ (one assignment). 
In the first case, the problem is trivially
unsolvable. In the second case, the clause can be removed. 

Similarly, if $ \Sigma^{-} $ contains a propositional blocking, 
then $ ( \Sigma^{+}, \Sigma^{-} ) $ is trivially unsolvable.
Such blockings originate from a $ B_j $ that is purely propositional,
or that has form $ \exists y \ P_{\lambda}( y ). $ 

A third important preprocessing step 
is \emph{removal of unit blockings}.
Let $ \sigma \in \Sigma^{-} $ be a blocking whose domain
is included in the domain of some clause $ c \in \Sigma^{+}. $ 
In that case, one can remove every substlet $ \overline{v}/\overline{c} $
from $ c, $ that has 
$ \bigcup \{ \overline{v}/\overline{c} \} \models \sigma. $ 
If this results in $ c $ being empty, then $ ( \Sigma^{+}, \Sigma^{-} ) $
trivially has no solution. 
If no $ \overline{v}/\overline{c} $ in any clause $ c \in \Sigma^{+} $ 
implies $ \sigma, $ 
then $ \sigma $ can be removed from $ \Sigma^{-}, $ 
because of Lemma~\ref{Lemma_conflict_in_clause}.

Applying removal of unit blockings to the translation of $ (I,\phi_2) $ 
above results in 
\[ \begin{array}{l}
       (X,Y) \ / \ ( x_0,x_1 ) \ | \ ( x_1,x_2 ) \\
       (Y,Z) \ / \ ( x_0,x_1 ) \ | \ ( x_1,x_2 ) \\
       \hline
   \end{array}
\]
It is worth noting that removal of propositional blockings can be viewed
as a special case of removal of unit blockings.

A GCSP can be solved by backtracking, similar to SAT solving.
A backtracking algorithm for GCSP can be either variable or clause based.
A variable based algorithm maintains a substitution
$ \Theta, $ which it tries to extend into a solution.
It backtracks by picking a variable $ v $ and
trying to assign it in all possible ways. It backtracks 
when $ \Theta $ makes a clause $ c \in \Sigma^{+} $ false, 
or a blocking $ \sigma \in \Sigma^{-} $ true. 

A clause based algorithm maintains a consistent set $ S $ of
substlets (whose union defines a substitution).
It backtracks by picking an undecided clause $ c \in \Sigma^{+}, $ 
and consecutively inserting all substlets that are consistent with
$ S $ into $ S. $ It backtracks when there is a clause
$ c $ all of whose atoms are in conflict with $ S, $ or
when $ \bigcup S $ makes a blocking true. 

Our experiments suggest that there is no significant difference
in performance, nor in programming effort, between the two variants. 
We will stick with clause based algorithms, because it seems that they
can be more easily combined with local consistency checking.

\section{Matching Using Refinement Stacks}
\label{Sect_unary_algo}

We first present the algorithm without learning,
and add learning in the next section. 
The algorithm that we present here is a simplification of the 
algorithm in \cite{deNivelle2016a}, which unfortunately could
not be made competitive. The previous algorithm
was based on a combination
of local consistency checking and lemma learning from conflicts.
Local consistency checking will be discussed in detail in
Section~\ref{Sect_local_consistency_checking}, because there 
is still a probability that it can be used as priori check. 

Local consistency checking means that one generates all subsets
of clauses up to some size $ S+1 $ and checks which substlets
can occur in solutions. Substlets that do not occur in any solution
of some subset, certainly do not occur in a solution of the complete GCSP.
In most instances, filtering with a small $ S, $ e.g. $ 1 $ or $ 2 $
results in an empty clause. The algorithm of \cite{deNivelle2016a}
was based on a combination of local consistency checking and decision.
It is discussed in more detail in Section~\ref{Sect_IJCAR2016}. 

The algorithm that we discuss in this section evolved from
\cite{deNivelle2016a}. The main differences are: 
Clauses are not checked against each other anymore. 
Instead, clauses
are checked only against the substitution in combination with blockings.
Secondly, learnt lemmas are flat, i.e. finite disjunctions of single
assignments to variables. In \cite{deNivelle2016a}, lemmas were
finite disjunctions of substlets. It turns out that this
simplification improves performance by a factor between 100 and 1000. 

In order to implement matching algorithms and local
consistency checking, one needs to be able to remove 
substlets from clauses, and reintroduce them during backtracking. 
We call the process of removing substlets from a clause \emph{refinement}.
Whenever a clause has been refined, it may trigger other
refinements. In the earlier algorithm, refinement of a clause
could directly trigger more refinements of other clauses.
In the current algorithm, refinement of a clause can only 
trigger possible extension of the substitution, but extension
of the substitution may still
trigger other clause refinements.
As a consequence, one needs to maintain a queue 
of recent refinements and use this queue to check which more clauses
can be refined. 
We introduce a data structure, called \emph{refinement stack} which
supports refinement of clauses, restoring during backtracking,
and keeping track of unchecked refinements.

\begin{defi}
   \label{Def_refinement_stack}
   A \emph{refinement} has form $ c \Rightarrow d, $ 
   where both $ c $ and $ d $ are clauses, and $ d $ is a 
   subclause of $ c. $ 

   \noindent
   A \emph{refinement stack} $ \overline{C} $ is a finite sequence
   of refinements $ c_i \Rightarrow d_i. $ 
   If there exists a $ j $ with $ i < j $ and $ c_i = c_j, $ then
   $ d_j $ must be a strict subclause of $ d_i. $ 

   \noindent
   For a clause $ c, $ if $ c_i \Rightarrow d_i $ is the last
   refinement with $ c = c_i $ occurring in $ \overline{C}, $ 
   we call $ d_i $ \emph{the current refinement of} $ c. $ 
   
   \noindent
   We define a predicate $ \alpha_i( \overline{C} ) $ that is true
   if $ c_i \Rightarrow d_i $ is the current refinement of $ c_i $ 
   in $ \overline{C}. $ This means that there is no $ j > i $ with
   $ c_j = c_i. $ 

   A refinement stack supports gradual refinement of clauses.
   If $ \alpha_i( \overline{C} ) $ is true, then clause
   $ d_i $ can be refined into $ d' $ by appending
   $ c_i \Rightarrow d' $ to $ \overline{C}. $

   In the new refinement stack 
   $ \overline{C}' = \overline{C} + ( c_{i} \Rightarrow d' ), $ 
   we have $ c_i = c_{n+1}, $ \ \ 
   $ \alpha_i( \overline{C}' ) $ is false, and 
   $ \alpha_{n+1}( \overline{C}' ) $ is true.  

   The \emph{size} $ \| \overline{C} \| $ of a refinement stack
   $ \overline{C} $ is defined as the total number of refinements 
   that occur in it,
   independent of the values of $ \alpha_i( \overline{C} ). $ 
\end{defi}

\noindent
The refinement stack is initialized with the refinements 
$ c \Rightarrow c, $ for each initial clause $ c. $
Refinement stacks can be efficiently implemented without need to copy
clauses by maintaining a stack of intervals of active substlets
in the initial clauses. 
A substlet can be disabled by swapping it with the last active
substlet in the interval, and decreasing the size of the interval by one. 
When the substlet is made active again, it is sufficient to restore
the interval, because the order of active substlets in a clause
does not matter. 
Refinement stacks support change driven inspection as well 
as backtracking. 

Change driven inspection
of clauses can be implemented by starting at position $ k=1. $ 
As long as $ k \leq \| \overline{C} \|, $ one first 
checks $ \alpha_k( \overline{C} ). $ If it is false, then
$ d_k $ is not the current version of $ c_k, $ and one can increase
$ k. $ 
If $ \alpha_k( \overline{C} ) $ is current, one can check if 
$ d_k $ triggers refinement of other clauses. 
If yes, the results are inserted at the end, so that 
they will be inspected at later time.
When one reaches $ k > \| \overline{C} \|, $ one has
reached a stable state. 

When some change involving a variable $ v $ takes place, one needs
to check which clauses may be affected by the change, so that
they can be refined. These are obviously the clauses that contain
$ v, $ but also the clauses that contain a variable occuring
in a blocking that contains $ v, $ since the algorithm
takes blockings into account, when refining. 
This gives rise to the following definition: 

\begin{defi}
   \label{Def_connected_variable}
   Let $ v,w $ be two variables.  We call $ v $ and $ w $ 
   \emph{connected} if $ v $ and $ w $ occur together in a blocking
   $ \sigma \in \Sigma^{-}. $ 
\end{defi}

\noindent
We define the search algorithm. We assume that 
propositional clauses and unit blockings have been removed from
$ ( \Sigma^{+}, \Sigma^{-} ). $
We assume that the substitution
$ \Theta $ is an ordered sequence (stack)
of assignments $ ( v_1/x_1, \ldots, v_s/x_s ). $ 

\begin{algo}
   \label{Algo_unary}
   We want to find a solution for $ ( \Sigma^{+}, \Sigma^{-} ). $
   Initially, set
   $ \Theta := \emptyset $ and $ \overline{C} := \emptyset. $  
   After that, for each $ k \ ( 1 \leq k \leq \| \Sigma^{+} \| ), $ 
   do the following:
   \begin{description}
   \item[PREPROC]
      Let $ c_k $ be the $ k $-th clause in $ \Sigma^{+}. $ 
      Append $ (c_k \Rightarrow c_k ) $ to $ \overline{C}. $ 
      For every variable $ v $ occurring in $ c_k, $ for which all
      substlets in $ c_k $ agree on the value of $ v, $ 
      let $ x $ be the agreed value. 
      \begin{itemize}
      \item
         If $ v \Theta $ is defined, and $ v \Theta \not = x, $ 
         then return $ \bot. $ 

      \item
         If $ v \Theta $ is undefined and there is a blocking
         $ \sigma $ containing $ v, $ s.t. 
         $ \Theta \cup \{ v/x \} \models \sigma, $ then 
         return $ \bot. $ Otherwise, append $ v/x $ to $ \Theta. $ 
      \end{itemize} 
   \end{description}

   \noindent
   After that, we call the main search algorithm 
   $ {\bf findmatch}( \overline{C}, \Theta, s, \Sigma^{-} ) $  
   with $ s = 1. $ 
   It either returns $ \bot, $ or
   it extends $ \Theta $ into a solution of
   $ ( \overline{C}, \Sigma^{-} ). $
   $ {\bf findmatch}( \overline{C}, \Theta, s, \Sigma^{-} ) $ 
   is defined as follows: 
   \begin{description}
   \item[FORW]
      As long as $ s \leq \| \Theta \|, $ let 
      $ v/x $ be the $ s $-th assignment of $ \Theta. $ 
      \begin{enumerate}
      \item
         For every $ ( c_i \Rightarrow d_i ) \in \overline{C} $
         which has $ \alpha_i( \overline{C} ) $ true, and which 
         either contains $ v $ itself, or a variable $ w $ that
         is connected to $ v, $ 
         let 
         \[ d' = \{ s \in d_i \ | \ s 
                    \mbox{ is not in conflict with } \Theta \}. \] 
         If $ d' = \emptyset, $ then return $ \bot. $
         Otherwise, let 
         \[ d'' = \{ s \in d' \ | \ 
            \mbox{there is no } \sigma \in \Sigma^{-}, \mbox{s.t. }
            \Theta \cup \{ s \} \models \sigma \}. \]
         If $ d'' = \emptyset, $ then return $ \bot. $ 
         Otherwise, if $ d'' \subset d_i, $ then
         \begin{enumerate}
         \item
            append $ ( c_i \Rightarrow d'' ) $ to $ \overline{C}. $
         \item
            For every variable $ v' $ occurring in $ d'', $ that
            is unassigned in $ \Theta, $ for which all substlets in 
            $ d'' $ agree on the assigned value, let
            $ x' $ be the agreed value. Append $ v'/x' $ to $ \Theta. $ 
         \end{enumerate}
      \item
         Set $ s = s + 1. $
      \end{enumerate}

   \item[PICK]
      Find an $ i $ with $ \alpha_i( \overline{C} ) $ true
      and $ \| d_i \| > 1. $
      If no such $ i $ exists, then $ \Theta $ is a solution.

      \noindent
      Otherwise, for every substlet $ \overline{v}_j/\overline{x}_j $ 
      in $ d_i, $ do the following: 
      \begin{enumerate}
      \item
         Append 
         $ c_i \Rightarrow ( \overline{v}_j/\overline{x}_j ) $ 
         to $ \overline{C}, $ and  
         extend $ \Theta $ with the unassigned 
         variables in $ \overline{v}_j / \overline{x}_j. $ 
         
      \item
         Recursively call $ {\bf findmatch}( \ \overline{C},
                             \Theta, s, \ \Sigma^{-} \ ). $ 
         If 
         $ \Theta $ was extended into a solution, then return $ \Theta. $ 
      \item
         Otherwise, restore $ \Theta $ and 
         $ \overline{C} $ to the sizes that they 
         had before (1). 
      \end{enumerate} 
      At this point, each of the recursive calls has returned 
      $ \bot. $ Return $ \bot. $ 
   \end{description} 
   \noindent
\end{algo}

\noindent
At {\bf FORW}, the algorithm attempts deterministic reasoning.
For every new assignment in $ \Theta, $ it is checked if it conflicts
with some substlets in some clause. Two types of conflicts are considered,
either the substlet contains an assignment that directly conflicts
with $ \Theta, $ or it contains an assignment that, together with 
$ \Theta, $ implies a blocking. 
As long as conflicts are found, the corresponding clauses are refined.
Refinement of a clause may result in $ \Theta $ being extended
({\bf FORW}~b), if the remaining substlets agree on an assignment. 
Extension of $ \Theta $ may result in further refinements of clauses.

\noindent
If {\bf FORW} failed to solve the problem, then at {\bf PICK}
a non-unit clause is picked, and non-deterministically
refined into a unit clause. This step requires backtracking.
It is important (for performance) to pick a clause of minimal length.

Main purpose of {\bf PREPROC} is to initialize the refinement 
stack $ \overline{C} $ with $ \Sigma^{+}. $ After that,
$ \Theta $ is initialized by looking for assignments that are
common to all substlets in some clause. If this results
in a conflict (either directly, or with a blocking), the problem
is rejected.

\noindent
Algorithm~\ref{Algo_unary} is similar to DPLL 
in that it tries to postpone backtracking as long
as possible by giving preference to deterministic extension. 
At {\bf FORW}, blockings are taken into account.
It is possible to implement {\bf FORW} without considering blockings.
In that case, it has to be checked, whenever the  
substitution is extended (at {\bf PICK~2} and at {\bf FORW~1b})
that the extended substitution does not imply a blocking. 
The given version performs better in experiments. 

In order to show that Algorithm~\ref{Algo_unary} is correct,
i.e. does not report false solutions, we have to 
show that all necessary checks are made. 

\begin{lem}
   \label{Lem_not_alone}\leavevmode
   \begin{enumerate}
   \item
      At points {\bf FORW} and {\bf PICK} of Algorithm~\ref{Algo_unary}, 
      there is no $ \sigma \in \Sigma^{-}, $ s.t. 
      $ \Theta \models \sigma. $ 
   \item
      At point {\bf PICK}, no refined clause $ d_i $
      contains a substlet that is in conflict with $ \Theta. $ 
   \end{enumerate}
\end{lem}
Initially, the preprocessor ensures that
there is no $ \sigma \in \Sigma^{-}, $ s.t. 
$ \Theta \models \sigma. $ 
When $ \sigma $ is extended in {\bf FORW~1b}, it has been checked
before that $ \Theta \cup \{ s \} $ does not imply a blocking,
for each of the substlets in $ d''. $ 
At point {\bf PICK}, \ {\bf findmatch} passed through 
{\bf FORW} which refined away all substlets that conflict with
$ \Theta. $ 

In the next section, we will extend Algorithm~\ref{Algo_unary}
with learning. This will prove completeness, because whenever
Algorithm~\ref{Algo_unary} does not find a solution, 
it will construct a lemma that proves that no lemma exists.

\section{Conflict Learning}
\label{Sect_conflict_learning}

\noindent
It is known from propositional SAT solving that 
conflict learning dramatically improves the performance
of SAT solvers (\cite{Handbook_sat:CDCL}).
The matching algorithm in the two-valued version
of {\bf Geo} (\cite{deNivelleMeng2007f}) was already
equipped with a primitive form of conflict learning.
Before releasing {\bf Geo}, we had experimented
with naive matching, the algorithm in \cite{GottlobLeitsch85},
and many ad hoc methods.
Matching with conflict learning is the only approach that
results in acceptable performance. Despite this, matching was still 
a critical operation in the last two-valued version of {\bf Geo}.
In the two-valued version of {\bf Geo}, lemmas had form
$ v_1/x_1, \dots, v_n/x_n \rightarrow \bot, $ i.e. they 
had form $ ( \overline{v} / \overline{x} ) \rightarrow \bot $ 
for a single substlet.

In \cite{deNivelle2016a} we proposed to replace the
lemmas of {\bf Geo}~2007 by arbitrary sets of substlets.
It is quite easy to see, that in general such a lemma 
can be in conflict with more substitutions than a 
lemma of the previous form. 
For example, if we assume that the domain is $ \{ X,Y,Z \} $
and the range $ \{ 0,1,2 \}, $ 
then $ (X,Y,Z)/(0,1,2) \rightarrow \bot $ rejects a single
substitution, while $ (X,Y,Z)/(0,1,2), \ (X,Y,Z) / (2,1,0) $
rejects 25 substitutions. 
Since in case of a conflict, one can always obtain a lemma
of the second form, it seemed that lemmas of the second
form should be preferred over lemmas of the first form. 

The latest version of {\bf Geo} see (\cite{CASC2016})
used the algorithm of \cite{deNivelle2016a} with 
lemmas of the unrestricted form above.
Although this matching algorithm 
performs better than matching in {\bf Geo}~2007, 
recent experiments have shown that it performs
significantly worse than some other approaches,
in particular translation to SAT and Algorithm~\ref{Algo_unary}
in combination with flat lemmas. 
Flat lemmas 
are lemmas of form $ v_1 \in V_1 \vee \cdots \vee v_n \in V_n. $ 
Surprisingly, Algorithm~\ref{Algo_unary} with
unrestricted lemmas performs several orders worse than
Algorithm~\ref{Algo_unary} with flat lemmas. 
This is surprising, because every general lemma can be 
flattened into a lemma of the second form by picking a single 
assignment from each substlet. The resulting lemma is obviously
less general than its original, non-flattened version. 
This loss of generality also applies to the reasoning rules that
we use on lemmas. If two substlets in two general lemmas
are in conflict, then their flattenings are not
necessarily in conflict. Conversely, whenever two flattened
substlets are in conflict, their original counterparts are.
This means that by using
flattened lemmas, one looses conflicts with substitutions,
and also resolution derivations involving lemma resolution. 
Despite this clever reasoning, the first columns of
Figure~\ref{Fig_unary_sat} of Section~\ref{Sect_experiments}
show that Algorithm~\ref{Algo_unary}
with flat lemmas performs approximately 200-400 times
worse than Algorithm~\ref{Algo_unary} with unrestricted lemmas.
One could assume that this is caused by the fact that
handling of unrestricted lemmas is more costly, and that 
their theoretical advantage is compensated by the increased
cost of their maintenance.
This assumption is rejected by
Figure~\ref{Fig_unary_sat}, because Algorithm~\ref{Algo_unary} 
with flattened lemmas is not only faster, but it also uses
less lemmas, typically by a factor 2-3. 
The only point where Algorithm~\ref{Algo_unary} with and without
flattening can diverge, is when a conflict lemma
rejects a substitution $ \Theta, $ and there exists more than one 
conflict lemma. Since both versions will prefer the shortest lemma, 
it must be due to the fact that flattening changes the relative
sizes of the lemmas.  

The outcomes of the experiments make it probable that the 
best approach to matching will be either Algorithm~\ref{Algo_unary} with
flat lemmas, or translation to SAT in combination with a SAT-solver, 
which we will describe in Section~\ref{Sect_SAT_trans}. 

From the practical point of view, the fact that the refining
algorithm in \cite{deNivelle2016a} turned out not competitive, is not a 
serious loss. Despite being elegant on paper, it was hard to implement. 
Implementation of Algorithm~\ref{Algo_unary} was much easier, 
and in the long term, it is better that the easier algorithm 
has the better performance.
Moreover, it is clear from Figure~\ref{Fig_unary_sat} that matching in
future versions of {\bf Geo} can be approximately $ 1000 $ times faster 
than it was at {\bf Geo}~2016c (\cite{CASC2016}). 

We will now introduce the flat lemmas, and prove that
Algorithm~\ref{Algo_unary} can always generate a flat conflict lemma.

\begin{defi}
   \label{Def_lemma}
   A \emph{lemma} is an object of form $ \{ v_1/V_1, \ldots, v_n/V_n \} $
   with $ n \geq 0. $ The $ v_i $ are variables, and the $ V_i $ 
   are finite sets of constants. 

   \noindent
   It is convenient to treat lemmas as total functions from variables
   to sets of constants.
   For a variable $ v $ and $ \lambda = \{ v_1/V_1, \ldots, v_n/V_n \}, $ \
   $ \lambda(v) $ is defined as $ \bigcup \{ V_i \ | \ v_i = v \}. $  

   \noindent
   Let $ \Theta $ be a substitution. We say that $ \Theta $ 
   \emph{makes} $ \lambda $ \emph{true} if 
   there exists a variable $ v $ in the domain of $ \Theta, $ for which
   $ v \Theta \in \lambda(v). $ 

   \noindent
   We say that $ \Theta $ makes $ \lambda $ false if all variables
   $ v $ for which $ \lambda(v) $ is nonempty, are in the domain
   of $ \Theta, $ and $ v \Theta \not \in \lambda(v). $ 
   In that case, we write $ \Theta \models \neg \lambda. $ 
\end{defi}

\begin{defi}
   \label{Def_valid_and_conflict}
   Let $ ( \Sigma^{+}, \Sigma^{-} ) $ be
   a GCSP.
   Let $ \lambda $ be a lemma.
   We say that $ \lambda $ is \emph{valid in}
   $ ( \Sigma^{+}, \Sigma^{-} ) $
   if every solution $ \Theta $ of $ ( \Sigma^{+}, \Sigma^{-} ) $
   makes $ \lambda $ true.

   \noindent
   For a given substitution $ \Theta, $ we call $ \lambda $
   a \emph{conflict lemma} if $ \lambda $ is valid and
   $ \Theta $ makes $ \lambda $ false. 
\end{defi}

\noindent
If $ \Theta $ is a substitution, and there exists a valid lemma
that is false in $ \Theta, $ then it is not possible to extend
$ \Theta $ into a solution of $ ( \Sigma^{+}, \Sigma^{-} ). $ 

In order to derive the conflict lemma, the following rules will
be used:

\begin{defi}
   \label{Def_reso_proj_sigma}
   Given a GCSP $ ( \Sigma^{+}, \Sigma^{-} ), $ we define
   the following derivation rules: 
   \begin{description}
   \item[RESOLUTION]
      Let $ \lambda_1, \ldots, \lambda_m $ be a sequence of lemmas.
      Let $ v $ be a variable. 
      Let $ V $ be the set of variables $ v, $ 
      for which one of the $ \lambda_j $ has $ \lambda(v) \not = \emptyset. $ 
      We define the $ v $-\emph{resolvent of}
      $ \lambda_1, \ldots, \lambda_m $ as
      \[ \{ v/\bigcap_{1 \leq j \leq m} \lambda_j(v) \} \cup
         \{  v'/\bigcup_{ 1 \leq j \leq m} \lambda_j(v') \ | \ 
             v' \in V \wedge v' \not = v \}. \] 
   \item[PROJECTION]
      Let $ c \in \Sigma^{+} $ be a clause, let $ \lambda $
      be a lemma. We call $ \lambda $ a \emph{projection} of $ c, $ 
      if every substlet $ (\overline{v}/\overline{x}) \in c $ 
      contains an assignment $ v/x, $ s.t. $ x \in \lambda(v). $ 
   \item[$\sigma$-RESOLUTION]
      Let $ \sigma \in \Sigma^{-} $ be a blocking.
      Write $ \sigma = \{ v_1/x_1, \ldots, v_n/x_n \} \ ( n > 0 ). $
      Let $ c_1, \ldots, c_n \in \Sigma^{+} $ be clauses, 
      chosen in such a way that 
      every variable $ v_i $ occurs in $ c_i. $ 
      For every $ c_i, $ let 
      \[ V_i = \{ \ x \ | \ c_i \mbox{ contains a substlet }
                          \overline{w}/\overline{y} \mbox{ which contains } 
                          v_i / x \mbox{ and } x \not = x_i \ \}. \]
      Then the lemma 
      \[ \{ v_1/V_1, \ldots, v_n/V_n \} \]
      is called a $ \sigma $-\emph{resolvent of} $ c_1, \ldots, c_n. $ 
   \end{description} 
\end{defi}

\noindent
The lemmas $ \{ \ x / \{ 1, 2, 3 \}, \ y / \{ 2,3 \} \ \} $ and
$ \{ \ x / \{ 3, 4 \}, \ y / \{ 3,4 \}, \ z / \{ 2 \} \ \} $ can resolve into
$ \{ \ x / \{ 3 \}, \ y / \{ 2,3,4 \}, \ z / \{ 2 \} \ \}. $ 
Given clauses
$ c_1 = \{ \ (x,y)/(1,2), \ (x,y)/(1,1), \ (x,y)/(3,3) \ \} $ and 
$ c_2 = \{ \ (y,z)/(1,2), \ (y,z)/(2,1) \ \}, $ and
a blocking 
$ (x,z)/(1,2), $ one can obtain the $ \sigma $-resolvent
$ \{ \ x / \{ 3 \}, \ z / \{ 1 \} \ \}. $ 
The lemma $ \lambda'_1 = \{ x / \{ 1,3 \} \} $ is a projection of
$ c_1. $ $ \lambda''_1 = \{ \ x / \{ 3 \}, \ y / \{ 1,2 \} \ \} $
is also a projection of $ c_1. $ 
It is easy to see that the reasoning rules are valid, 
which implies that every lemma that has been
obtained by repeated application from the original
clauses in $ \Sigma^{+} $ and blockings in $ \Sigma^{-}, $ is valid.

\begin{lem}
   \label{Lem_sigma_resolvent_false} 
   Let $ ( \Sigma^{+}, \Sigma^{-} ) $ be a GCSP. 
   Let $ \Theta $ be an interpretation. Let $ \sigma \in \Sigma^{-} $
   be a blocking for which $ \Theta \models \sigma. $ 
   Let $ \lambda $ be a $ \sigma $-resolvent of $ \sigma. $ 
   Then $ \Theta $ makes $ \lambda $ false. 
\end{lem}
\begin{proof}
   Write $ \sigma = \{ v_1/x_1, \ldots, v_n/x_n \}. $
   Let $ c_1, \ldots, c_n \in \Sigma^{+} $ be the clauses
   that were used in the construction of $ \lambda. $ 
   Because $ \Theta \models \sigma, $ we know that
   for every $ i \ ( 1 \leq i \leq n ), $ we have $ v_i \Theta = x_i. $ 
   From the construction of the $ V_i, $ it follows that $ x_i \not \in V_i. $ 
   Because the variables $ v_i $ are pairwise distinct, we 
   have $ V_i = \lambda(x_i). $ It follows that $ x_i \not \in \lambda(v_i). $
   For all other variables $ v $ that do not occur in $ \sigma, $ 
   we have $ \lambda(v) = \emptyset. $ 
   We can conclude that if
    $ \lambda(v) $ is non-empty, then $ v $ equals one of the $ v_i, $ 
   and we have $ v_i \Theta \not \in \lambda(v_i). $ 
\end{proof}
The following lemma states that substlets that are switched off, 
were switched off because they conflict $ \Theta, $ possibly with 
help of a blocking.

\begin{lem}
   \label{Lemma_weak_conflict}
   At every moment during Algorithm~\ref{Algo_unary}, 
   for every refinement $ (c_i \Rightarrow d_i) \in \overline{C}, $ the 
   following holds:
   If $ s \in ( d_i \backslash c_i ), $ then
   either 
   $ s $ is in conflict with $ \Theta, $ or
   $ \Theta \cup \{ s \} \models \sigma, $ for a $ \sigma \in \Sigma^{-}. $ 
\end{lem}
\begin{proof}
   There are two points at which refinement can take place, 
   {\bf PICK~1} and
   {\bf FORW~1a}. At {\bf PICK~1}, clause $ c_i $ is refined into
   $ \overline{v}_j / \overline{x}_j, $ after which $ \Theta $ is
   extended with $ \overline{v}_j / \overline{x}_j. $
   If some substlet 
   $ s $ occurs in $ c_i \backslash \{ \overline{v}_j / \overline{x}_j \}, $
   then either $ s \in c_i \backslash d_i, $ or
   $ s \in d_i \backslash \{ \overline{v}_j / \overline{x}_j \}. $
   In the first case, the desired property is inherited from the 
   previous state, because it is an invariant.
   In the second case, because $ \Theta $ 
   is extended by $ \overline{v}_j / \overline{x}_j $ at the same time,
   we can apply Lemma~\ref{Lemma_conflict_in_clause}. 

   At {\bf FORW~1a}, if $ s \in d'' \backslash c_i, $ then either
   $ s \in d_i \backslash c_i, \ \ 
     s \in d' \backslash d_i, $ or 
   $ s \in d'' \backslash d'. $ 
   In the first case, 
   the desired property is inherited from the previous state. 
   In the second case, it follows from the construction of $ d', $ that
   $ s $ was in conflict with $ \Theta. $ 
   In the third case, it follows from the construction of $ d'', $ that
   there is a $ \sigma \in \Sigma^{-}, $ for which 
   $ \Theta \cup \{s \} \models \sigma. $
\end{proof}

\noindent
The following property is the essential property, for proving
that Algorithm~\ref{Algo_unary} can always return a conflict lemma. 
\begin{lem}
   \label{Lem_unary_derive}
   Let $ ( \Sigma^{+}, \Sigma^{-} ) $ be a GCSP. 
   Let $ c \in \Sigma^{+} $ be a clause. 
   Let $ \Theta $ be a substitution. Let $ \Lambda $ be a set 
   of lemmas. Assume that there is no $ \sigma \in \Sigma^{-}, $ s.t.
   $ \Theta \models \sigma, $ and no $ \lambda \in \Lambda, $ s.t.
   $ \Theta $ makes $ \lambda $ false.
   Assume that for every substlet $ s \in c, $ either
   \begin{enumerate}
   \item
      $ s $ is in conflict with $ \Theta, $ 
   \item
      $ \Theta \cup \{ s \} \models \sigma, $ 
      for a $ \sigma \in \Sigma^{-}, $ or
   \item 
      $ \Theta \cup \{ s \} $ makes a $ \lambda \in \Lambda $ false.
   \end{enumerate}
   Then it is possible to derive a conflict lemma for $ \Theta $ 
   from $ \Sigma^{+} $ and $ \Lambda, $ by applying the rules in 
   Definition~\ref{Def_reso_proj_sigma}.
\end{lem}
\begin{proof}
   We first remove {\bf (2)} by means of $ \sigma $-resolution.
   We will add the resulting $ \sigma $-resolvents to $ \Lambda.$
   For every $ s \in c, $ for which $ {\bf (1),(3)} $ do not apply, 
   $ {\bf (2)} $ must apply. 
   Write $ \sigma = \{ v_1 /x_1, \ldots, v_n/x_n \}. $ 
   Since $ ( \Sigma^{+}, \Sigma^{-} ) $ is range restricted,
   we can find clauses 
   $ c_1, \ldots, c_n \in \Sigma^{+}, $ s.t.
   each $ v_i $ occurs in $ c_i. $  
   We now can construct the $ \sigma $-resolvent. Write
   $ \lambda $ for the resulting lemma. 
   It follows from Lemma~\ref{Lem_sigma_resolvent_false} that 
   $ \lambda $ is false in $ \Theta. $ We can add $ \lambda $ to $ \Lambda. $ 
   At this point, we have for every $ s \in c, $ either
   $ {\bf (1)} $ or $ {\bf (3)}. $ 
   The rest of the proof is Lemma~\ref{Lem_unary_derive_recurse}.  
\end{proof}

\begin{lem}
   \label{Lem_unary_derive_recurse}
   Let $ ( \Sigma^{+}, \Sigma^{-} ) $ be a GCSP.
   Let $ c \in \Sigma^{+} $ be a clause.
   Let $ \Theta $ be a substitution. Let $ \Lambda $ be a set
   of lemmas. Assume that there is no $ \lambda \in \Lambda, $ s.t. 
   $ \Theta $ makes $ \lambda $ false.
   Assume that for every substlet $ s \in c, $ either
   \begin{enumerate}
   \item
      $ s $ is in conflict with $ \Theta, $ or 
   \item
      $ \Theta \cup \{ s \} $ makes a $ \lambda \in \Lambda $ false.
   \end{enumerate}
   Then it is possible, using the rules in
   Definition~\ref{Def_reso_proj_sigma}, to obtain a conflict lemma for
   $ \Theta $ from $ c $ and $ \Lambda. $ 
\end{lem}
\begin{proof}
   We prove the lemma by induction on the number of unassigned variables
   in $ c. $ 
   Let $ c_1 $ be the part of $ c $ to which {\bf (1)} applies,
   and let $ c_2 = c \backslash c_1. $ 
 
   Since each $ s \in c_1 $ is in conflict with $ \Theta, $  
   one can obtain a projection $ \mu_1 $ of $ c_1 $ by picking
   from each $ s \in c_1 $ an assignment $v/x $ for which
   $ v \Theta $ is defined and $ v \Theta \not = x. $ 
   By construction, $ \mu_1 $ will be false in $ \Theta. $ 

   If there are no unassigned variables in $ c, $ then 
   $ c_2 $ must be empty. This means that $ \mu_1 $ is a projection
   of $ c $ and false in $ \Theta, $ so we are done.

   Otherwise, select a $ v $ in $ c $ that is unassigned by
   $ \Theta. $ 
   Let $ V $ be set of values that are assigned to $ v $ 
   by the substlets in $ c_2. $ 
   Define $ \mu_2 = \{ v / V \}. $ Clearly, $ \mu_2 $ is a 
   projection of $ c_2, $ and $ \mu = \mu_1 \cup \mu_2 $ is a projection
   of $ c. $ 
   For each value $ x \in V, $ define $ \Theta_x = \Theta \cup \{ v/x \}. $ 
   If there is no $ \lambda \in \Lambda, $ 
   that is false in $ \Theta_x, $ then $ c, \ \Theta_{x}, \ \Lambda $
   still satisfy the conditions of Lemma~\ref{Lem_unary_derive_recurse}. 
   Moreover, since $ \Theta_x $ contains an assignment to $ v, $ 
   the number of unassigned variables in $ c $ has decreased by one.
   This means that we can assume, by induction, that we can derive
   a lemma $ \lambda_x $ that is false in $ \Theta_{x}. $  
   If $ \lambda_{x} $ is also a conflict lemma of $ \Theta, $ 
   we have completed the proof. 
   Otherwise, we can assume that $ \lambda_x $ is added to $ \Lambda. $ 

   At this point, $ \Lambda $ contains a conflict lemma 
   $ \lambda_x $ for every
   $ \Theta \cup \{ v/x \} $ with $ x \in V. $ 
   Let $ \lambda $ be the $ v $-resolvent of the projection
   $ \mu $ constructed above, and the $ \lambda_x, $ i.e. 
   \[ \lambda = 
      \{ \ v / ( \ \mu(v) \cap \bigcap_{x \in V} \lambda_x(v) \ ) \ \} \cup
      \{ \ v' / ( \ \mu(v') \cup \bigcup_{x \in V} \lambda_x(v') \ ) 
                \ | \ v' \not = v \ \}. \]
   In order to show that $ \lambda $ is false in $ \Theta, $ 
   we have to show that for every variable $ v', $ for which
   $ \lambda(v') \not = \emptyset, \ \ v' \Theta $ is defined,
   and $ v' \Theta \not \in \lambda(v'). $ 
   \begin{itemize}
   \item
      For $ v, $ we just show that $ \lambda(v) = \emptyset. $ 
      We have $ \mu(v) = \mu_2(v), $ because $ \mu_1(v) = \emptyset. $ 
      It follows from the fact that $ v $ is undefined in $ \Theta, $ 
      and $ \mu_1 $ is false in $ \Theta. $ 
      For each $ x \in \mu_2(v), $ we know that $ \lambda_{x} $ is
      false in $ \Theta \cup \{ v/x \}, $ which implies that
      $ x \not \in \lambda_x(v). $ This implies that $ x $ is not
      in the intersection of all $ \lambda_x(v), $  
      which in turn implies $ \mu(v) $ and $ \bigcap_{x \in V} \lambda_x(v) $ 
      have no elements in common. 
   \item
      If $ v' \not = v $ and $ \lambda(v') \not = \emptyset, $ 
      then either $ \lambda_x(v') \not = \emptyset, $ for an $ x \in V, $ 
      or $ \mu(v') \not = \emptyset. $ In the first case, it follows
      from the fact that $ \lambda_x $ is false in $ \Theta \cup \{ v/x \} $ 
      and $ v' \not = v, $ 
      that $ v' \Theta $ is defined. In the second case, 
      we know that $ \mu_2(v') $ only assigns to $ v, $ so that 
      $ \mu_1(v') \not = \emptyset. $ Since we know that $ \mu_1 $ is false
      in $ \Theta, $ we know $ v' \Theta $ is defined.

      At this point, we are certain that $ v' \Theta $ is defined,
      so that we can start showing that 
      $ v' \Theta \not \in \mu(v') \cup \bigcup_{x \in V} \lambda_x(v'). $
      If $ v' \Theta \in \mu(v'), $ then, because $ \mu_2 $ 
      only assigns to $ v, $ we have $ v' \Theta \in \mu_1(v') $ 
      This is impossible because $ \mu_1 $ is false in $ \Theta. $ 

      We can also not have $ v' \Theta \in \lambda_x(v'), $ 
      for any $ x \in V, $ because this would imply that
      $ v' ( \Theta \cup \{ v/x \} ) \in \lambda_x(v'), $ which 
      contradicts the fact that $ \lambda_x $ is false in
      $ \Theta \cup \{ v/x \}. $  \qedhere
   \end{itemize}
\end{proof} 

\noindent
At this point, it is straightforward to prove that
Algorithm~\ref{Algo_unary} can always derive a conflict lemma. 
There are two points in Algorithm~\ref{Algo_unary} where the
substitution is extended. We show for both points that it is possible
to obtain a conflict lemma when the substitution is restored. 

\begin{description} 
\item[FORW 1b]
   The substitution $ \Theta $ is extended by the common 
   assignments in $ d''. $ 
   Since the extension of $ \Theta $ had a conflict lemma,
   we know that for each $ s \in d'', \ \ \Theta \cup \{ s \} $ 
   has a conflict lemma. 
   It follows from Lemma~\ref{Lemma_weak_conflict} that for every 
   substlet $ s $ in $ d'' \backslash \{ c_i \}, $ either $ s$ is in conflict
   with $ \Theta, $ or $ \Theta \cup \{ s \} $ implies $ \sigma, $ 
   for a blocking $ \sigma \in \Sigma^{-}. $ 
   From Lemma~\ref{Lem_not_alone}, we know that
   there is no $ \sigma \in \Sigma^{-}, $ s.t. $ \Theta \models \sigma. $
   It follows that we can apply Lemma~\ref{Lem_unary_derive} with
   $ \Lambda = \{ \lambda \} $ to obtain a conflict lemma
   for $ \Theta. $ 
\item[PICK]
   Let $ c_i \Rightarrow d_i $ be the refinement that was selected
   by PICK.
   Let $ \Lambda $ be the set of conflict lemmas that were returned
   by the recursive calls of {\bf findmatch}.
   If there is a $ \lambda \in \Lambda $ that is false in $ \Theta, $ 
   we can return $ \lambda. $
   Otherwise, we know that no $ \lambda \in \Lambda $ is false
   in $ \Theta. $ 
   From Lemma~\ref{Lem_not_alone}, we know that
   there is no $ \sigma \in \Sigma^{-}, $ s.t. $ \Theta \models \sigma. $ 
   
   By Lemma~\ref{Lemma_weak_conflict}, every substlet
   $ s \in ( c_i \backslash d_i ), $ 
   is either in conflict with $ \Theta, $ or there exists a 
   $ \sigma \in \Sigma^{-}, $ s.t. $ \Theta \cup \{ s \} \models \sigma. $ 
   This implies that we can apply Lemma~\ref{Lem_unary_derive} to obtain a  
   conflict lemma of $ \Theta. $  
\end{description}

In an implementation of Algorithm~\ref{Algo_unary}, there is no need to 
follow the rules of Definition~\ref{Def_reso_proj_sigma} carefully, because 
the conflict lemma can be constructed immediately from the premisses 
of Lemma~\ref{Lem_unary_derive}. 

In order to make Algorithm~\ref{Algo_unary} reuse conflict lemmas, one 
has to add before {\bf FORW~1}: If there is a 
$ \lambda \in \Lambda $ containing variable $v, $ s.t.
$ \Theta $ makes $ \lambda $ false, then return $ \lambda. $ 

Integrating lemmas into the refining step of {\bf FORW}~1 seems difficult, 
because the notion of connection (Definition~\ref{Def_connected_variable}) 
must be extended to include 
`$ v $ and $ w $ occur together in a lemma $ \lambda \in \Lambda. $' 
Currently we don't know how to efficiently enumerate variables that are 
connected through a lemma. 

\section{Matching Based on Local Consistency Checking}
\label{Sect_IJCAR2016}

We will discuss the matching algorithm of \cite{deNivelle2016a}.
Its performance turned out not competitive, so we will omit most of 
the details, in particular the completeness proofs for learning.  
The algorithm is based on the fact that local consistency
checking rejects a large percentage of GCSPs
without backtracking. 

Local consistency checking is the following procedure:
For every clause $ c = \{ s_1, \ldots, s_n \} \in \Sigma^{+}, $ check,
for all sets of clauses $ C $ with size $ S \geq 1, $ 
if $ \{ s_i \} \cup C $ has a solution. If not, then
remove $ s_i $ from $ c. $ 
Keep on doing this, until no further changes are possible or a clause 
has become empty. The procedure is described in detail
in Section~\ref{Sect_local_consistency_checking}. 
Local consistency checking with small $ S $ rejects a large
percentage of instances without backtracking. 
It therefore seemed reasonable to 
combine local consistency checking with backtracking in the
following way: 
\begin{description}
\item[FILTER]
   Apply local consistency checking. If this results in an empty
   clause, then backtrack to the last decision. If there are no decisions
   left, then report failure. 
\item[DECIDE]
   If every clause has become unit, then report a solution. 
   Otherwise, pick a non-unit clause, and replace it by a singleton
   consisting of one of its substlets. 
   Continue at {\bf FILTER}.
   If this results in an empty clause, then backtrack through the 
   remaining substlets of the clause. 
\end{description}
The assumption was that local consistency checking could play
the same role as unit propagation in DPLL, and that local consistency
checking would be equally effective on the subproblems obtained 
during backtracking, as on the initial problem. 
This assumption turned out false. 
In \cite{deNivelle2016a}, the algorithm is described for
$ S = 1, $ but we have implemented it for arbitrary $ S \geq 1. $
Note that a size of $ S $ means that $ \| C \| = S, $ so that
$ \{ c_i \} \cup C $ has size $ S+1. $ 
Performance results are presented in 
Figure~\ref{Fig_tyebye_z_maslem} in Section~\ref{Sect_experiments}.
It can be seen that 
$ S > 1 $ does not perform better than $ S = 1. $ It rarely
creates less lemmas, and it usually costs more time. 
 
The main observation to be made is that the algorithm is
not close to being competitive against 
Algorithm~\ref{Algo_unary} with flat lemmas,
or translation to SAT. In addition to that, 
it turned out rather unpleasant
to implement, much harder than Algorithm~\ref{Algo_unary}. Especially
$ S > 1 $ is difficult to handle, because the resolution rules
for obtaining lemmas become rather complicated. This does not
only apply to the implementation, but also to the theoretical
description. 

We define the lemmas that were used by 
the matching algorithm, and the reasoning rules
that it uses. 
A clause can be viewed as a special form of lemma 
in which the substlets have the same domain. 

\begin{defi}
   \label{Def_full_lemma} 
   A \emph{lemma} is a finite set of substlets, possibly
   with different domains. 

   If $ \lambda $ is a lemma, and $ \Theta $ a substitution,
   then $ \Theta $ makes $ \lambda $ true if there
   is a substlet $ ( \overline{v}/\overline{x} ) \in \lambda, $ 
   s.t. $ \Theta $ makes $ \lambda $ true.
   $ \Theta $ makes $ \lambda $ false if every
   substlet $ ( \overline{v}/\overline{x} ) \in \lambda $ 
   is in conflict with $ \Theta. $ 
   We say that $ \lambda $ is \emph{valid} relative to
   $ ( \Sigma^{+}, \Sigma^{-} ), $ if
   $ \Theta $ is true in every solution $ \Theta $ of 
   $ ( \Sigma^{+}, \Sigma^{-} ). $ 
   We call $ \lambda $ a \emph{conflict lemma} if $ \lambda $ is
   false in the current $ \Theta $ and valid
   $ ( \Sigma^{+}, \Sigma^{-} ). $ 
\end{defi}
Learning was based on the following resolution rules: 

\begin{defi}
   \label{Def_full_conflict_resolution}
   Let $ \lambda_1 $ and $ \lambda_2 $ be lemmas. 
   Let $ \mu_1 \subseteq \lambda_1, $ 
   and let $ \mu_2 \subseteq \lambda_2. $ 
   Assume that every $ s_1 \in \mu_1 $ is in conflict with
   every $ s_2 \in \mu_2. $ 
   Then $ ( \lambda_1 \backslash \mu_1 ) \cup 
          ( \lambda_2 \backslash \mu_2 ) $ 
   is a resolvent of $ \lambda_1 $ and $ \lambda_2. $ 
\end{defi}

One can resolve 
$ \lambda_1 = \{ \ (x,y)/(1,2), \ (x,y)/(1,1), \ (x,y)/(3,3) \ \} $ with  
$ \lambda_2 = \{ \ (y,z)/(1,2),$ $ \ (y,z)/(2,1) \ \} $ based on
$ \mu_1 = \{ \ (x,y)/(1,2), \ (x,y)/(3,3) \ \}, $ and 
$ \mu_2 = \{ \ (y,z)/(1,2) \ \}. $ 
The resolvent is $ \{ \ (x,y)/(1,1), \ (y,z)/(2,1) \ \}. $ 

\begin{defi}
   \label{Def_full_sigma_resolution} 
   Let $ \sigma \in \Sigma^{-} $ be a blocking. 
   Let $ c_1, \ldots, c_n \in \Sigma^{+} $ be a sequence
   of clauses containing all variables of $ \sigma. $
   For each $ c_i, $ let 
   $ \rho_i = \{ \ s \in c_i \ | \ s \mbox{ is in conflict with } \sigma \ \}. $ 
   Then $ \rho_1 \cup \cdots \cup \rho_n $ is a \emph{$ \sigma $-resolvent}
   of $ c_1, \ldots, c_n. $ 
\end{defi}

% \noindent
Using $ \lambda_1, \lambda_2 $ given above, and blocking 
$ (x,z)/(1,2), $ one can obtain the $ \sigma $-resolvent
$\{ \ (x,y) / (3,3), \ (y,z)/(2,1) \ \}.$

It is easy to see that both conflict resolution and $ \sigma $-resolution 
are valid reasoning rules, which implies that every lemma that was
derived by repeated application of resolution from the original
clauses in $ \Sigma^{+}, $ is valid.  

In \cite{deNivelle2016a}, it was shown that a matching algorithm
using $ S = 1 $ can always obtain a conflict lemma using
resolution and $ \sigma $-resolution. For $ S > 1, $ an additional
rule, called \emph{product resolution}, is required. 
Results are listed in Figure~\ref{Fig_tyebye_z_maslem}.

After observing that Algorithm~\ref{Algo_unary} improves by a 
factor $ 500 $ when lemmas are flattened, we tried the same
with the refining algorithm. Whenever a new lemma is
derived, the assignments that do not contribute to conflicts are
removed from the substlets. Different from Algorithm~\ref{Algo_unary}, 
this does not necessarily lead to a lemma consisting only of 
single-assignment substlets, but in most cases it does. 
Surprisingly, this has a strong, negative impact on the performance.

\section{Translation to SAT}
\label{Sect_SAT_trans}

\noindent
Translating an instance of the matching problem to SAT is easy, 
and modern SAT solvers have become very effective. 
As a consequence, translation to SAT should be attempted.
In this section, we give two methods of translating GCSP into SAT.
The translations are not complicated, 
and MiniSat \cite{Minisat2004} performs rather well on the results
of the translations. Results are listed in the last two columns of
Figure~\ref{Fig_unary_sat} and in Figure~\ref{Fig_unary_sat2}
in Section~\ref{Sect_experiments}.
The results suggest that translation to SAT has a performance
that is comparable with Algorithm~\ref{Algo_unary}.

In our first translation only substlets are translated. We assign 
propositional variables to the substlets, specify that at least 
one substlet from each clause has to be selected, and list the conflicts 
between the substlets.

\begin{defi}
   \label{Def_trans_general}
   We assume a general mapping $ [ \ ] $ that transforms mathematical
   objects into distinct propositional variables.
\end{defi}

\begin{defi}
   \label{Def_trans_SAT1}
   Let $ ( \Sigma^{+}, \Sigma^{-} ) $ be GCSP. 
   The translation into propositional logic has form
   $ ( A, P ), $ 
   where $ A $ is a set of atoms, and $ P $ is a set of 
   clauses over $ A. $
   Assume that the GCSP has form $ ( \Sigma^{+}, \Sigma^{-} ), $ 
   assume that $ \Sigma^{+} $ contains $ n $ clauses, 
   and write $ \{ s_{i,1}, \ldots, s_{i,k_i} \} $ for the $i$-th clause
   of $ \Sigma^{+}. $ 

   The set of atoms is defined as 
   $ A = \{ [ s_{i,j} ] \ | \ 1 \leq i \leq n, \
                              1 \leq j \leq k_n \}. $ 
   The clause set $ P $ is defined as follows: 
   \begin{enumerate}
   \item
      For every $ c_i = \{ s_{i,1}, \ldots, s_{i,k_i} \} \in \Sigma^{+} \
      ( 1 \leq i \leq n ), $ 
      the propositional clause set $ P $ contains the propositional clause
      $ \{ \ [s_{i,1}], \ldots, [s_{i,k_i}] \ \}, $ 
      and for every $ j_1,j_2 \ ( 1 \leq j_1 < j_2 \leq k_i ) $ the
      clause $ \{ \ \neg [ s_{i,j_1} ], \ \neg [ s_{i,j_2} ] \ \}. $ 
   \item
      For every pair of distinct clauses 
      $ c_{i_1}, c_{i_2} \in \Sigma^{+} $ 
      that share a variable, for every substlet
      $ s \in c_{i_1}, \ \ P $ contains the clause
      \[ \{ \ \neg [s] \ \} \cup \{ \ [ s' ] \in c_{i_2} \ | \ 
          s' \in c_{i_2}, \mbox{ and }
          s' \mbox{ is not in conflict with } s \ \}. \]
   \item
      For every blocking $ \sigma \in \Sigma^{-}, $ 
      we assume that there is a way of selecting a 
      most suitable subset $ C_{\sigma} $ of $ \Sigma^{+} $ that
      contains all variables of $ \sigma. $ 
      Then $ P $ contains the clause 
      \[ \{ \ [s] \ | \ \exists c \in C_{\sigma}, \mbox{ s.t. } 
                     s \in c 
                     \mbox{ and } s 
                     \mbox{ is in conflict with } \sigma \ \}. \] 
   \end{enumerate} 
\end{defi}
The first part specifies that exactly 
one substlet must be selected from each $ c \in \Sigma^{+}. $ 
The second part specifies
that if one selects a substlet $ s $ from $ c_{i_1}, $ one has
to select a substlet $ s' $ from $ c_{i_2} $ that is not in conflict
with $ s. $ 
The third part of Definition~\ref{Def_trans_SAT1} can be
viewed as an application of
$ \sigma $-RESOLUTION (Definition~\ref{Def_full_sigma_resolution}).

The second translation differs from the first translation in
the fact that it does not only translate substlets, 
but also variable assignments. In addition
to the substlets, it assigns propositional variables to 
variable assignments $ v/x. $ It specifies the dependencies 
between substlets and variable assignments. Instead of 
relying on $ \sigma $-RESOLUTION, blockings can be specified
directly in terms of the forbidden variable assignments. 

\begin{defi}
   \label{Def_trans_SAT2}
   Let $ ( \Sigma^{+}, \Sigma^{-} ) $ be a GCSP. 
   Write $ \Sigma^{+} = \{ c_1, \ldots, c_n \}. $ 
   Write each $ c_i $ in the form $ \{ s_{i,1}, \ldots, s_{i,k_i} \}. $ 
   The translation to propositional logic has form
   $ ( A, P ), $ where 
   $ A $ is the set of atoms used in the translation,
   and $ P $ is the set of clauses.
   The set of atoms $ A $ is defined as 
   \[ \{ \ [s_{i,j}] \ | \ 1 \leq i \leq n, \ 
                           1 \leq j \leq k_i \ \} \cup
      \{ \ [v/x] \ | \ 
           (v/x) \mbox{ occurs in a substlet in } \Sigma^{+} \ \}. \]
   The set of propositional clauses $ P $ is obtained as follows: 
   \begin{enumerate}
   \item
      For every clause $ c_i, \ (1 \leq i \leq n), $ 
      clause set $ P $ contains the clause
      $ \{ \ [ s_{i,1} ], \ldots, [ s_{i,k_i} ] \ \}. $ 
   \item
      For every substlet 
      $ s_{i,j} $ with $ 1 \leq i \leq n, \ 1 \leq j \leq k_{i}, $ 
      for every assignment $ v/x $ that occurs in $ s_{i,j}, $ 
      clause set $ P $ contains the clause
      $ \{ \ \neg [ s_{i,j} ], \ [ v/x ] \ \}. $ 
   \item
      For every variable $ v $ that occurs in $ \Sigma^{+}, $ 
      for every two distinct values $ x_1, x_2, $ s.t.
      $ v/x_1 $ and $ v/x_2$ occur somewhere in substlets in $ \Sigma^{+}, $ 
      clause set $ P $ contains the clause
      $ \{ \ \neg [ v/x_1], \ \neg [ v/x_2 ] \ \}. $ 
   \item
      For every blocking $ \sigma \in \Sigma^{-}, $ 
      if every $ (v/x) \in \sigma $ occurs somewhere in a clause
      in $ \Sigma^{+}, $ then 
      clause set $ P $ contains the clause
      $ \{ \ \neg [ v/x ] \ | \ (v/x) \in \sigma \ \}. $ 
      If some $ (v/x) \in \sigma $ does not occur in $ \Sigma^{+}, $ then
      $ \sigma $ is impossible, and there is no need to generate a clause
      for it. 
   \end{enumerate}
\end{defi}

\noindent
We show correctness of Definition~\ref{Def_trans_SAT2}. 
If $ ( \Sigma^{+}, \Sigma^{-} ) $ has a solution $ \Theta, $ 
one can define a satisfying interpretation $ I $ for $ (A,P) $ as
follows: 
\begin{itemize}
\item
   For $ 1 \leq i \leq n, \ 1 \leq j \leq k_i, $
   set $ I( \ [ s_{i,j} ] \ ) = {\bf t} $ iff $ \Theta \models s_{i,j}. $
\item
   For every assignment $ v/x $ occurring in a substlet $ s $ 
   occurring in a clause $ c_i, $ set 
   $ I( \ [v/x] \ ) = {\bf t} $ iff $ v \Theta = x. $ 
\end{itemize}
It is easily checked that $ I $ makes all clauses in
Definition~\ref{Def_trans_SAT2} true.

For the other direction, assume that $ (A,P) $ has a satisfying 
interpretation $ I. $
Define $ \Theta = \{ \ (v/x) \ | \ I( \ [v/x] \ ) = {\bf t} \ \}. $
By part~4, $ \Theta $ does not contain conflicting assignments. 
By part~1 and part~2, $ \Theta $ contains an assignment for 
every variable occurring in $ \Sigma^{+}. $
Because of part~3, $ \Theta $ does not imply a blocking 
$ \sigma \in \Sigma^{-}. $
By part~1 and part~2, every $ c_i \in \Sigma^{+} $ contains one
substlet that agrees with $ \Theta. $ 

\noindent
We end the section with an example of both translations: 

\begin{exa}
   \label{Example_SAT1_SAT2} 
   We will translate the following GCSP. As usual, $ \Sigma^{+} $ 
   and $ \Sigma^{-} $ are separated by a horizontal bar.
   \[
      \begin{array}{l}
         (X,Y) \ / \ ( 0,1 ) \ | \ ( 1,0 ) \\
         (Y,Z) \ / \ ( 0,0 ) \ | \ ( 0,1 ) \ | \ ( 1,0 ) \\
         \hline
         (X,Z) \ / \ ( 0,0 ) \\
         (X,Z) \ / \ ( 1,1 ) \\
      \end{array}
   \]
   $ \Sigma^{+} $ alone has three solutions:
   \[ \begin{array}{l}
         \Theta_1 = \{ \ X := 0, \ Y := 1, \ Z := 0 \ \}, \\
         \Theta_2 = \{ \ X := 1, \ Y := 0, \ Z := 0 \ \}, \\ 
         \Theta_3 = \{ \ X := 1, \ Y := 0, \ Z := 1 \ \}. \\
      \end{array}
   \]
   The first solution is blocked by $ (X,Z)/(0,0), $ the
   third solution is blocked by $ (X,Z)/(1,1), $ so that 
   only $ \Theta_2 $ is a solution of the complete GCSP.
   Assume that
   \[ [ (X,Y)/(0,1) ] = 1, \ \ 
      [ (X,Y)/(1,0) ] = 2, \ \ 
      [ (Y,Z)/(0,0) ] = 3, \ \ 
      [ (Y,Z)/(0,1) ] = 4, \ \ 
      [ (Y,Z)/(1,0) ] = 5. \]
   Definition~\ref{Def_trans_SAT1} constructs the following translation:  
   \[
       {\rm Part}~1: 
       \left ( 
       \begin{array}{ccc}
          1 & 2 \\
          3 & 4 & 5 \\ 
          -1 & -2 \\
          -3 & -4 \\
          -3 & -5 \\
          -4 & -5 \\
       \end{array}
       \right ) 
       \ \ {\rm Part}~2: 
       \left ( 
       \begin{array}{ccc}
          -1 & 5 \\
          -2 & 3 & 4 \\
          -3 & 2 \\
          -4 & 2 \\
          -5 & 1 \\
       \end{array}
       \right ) 
       \ \ 
       {\rm Part}~3: 
       \left ( 
       \begin{array}{ccc} 
          1 & 4 \\
          2 & 3 & 5 \\
      \end{array}
      \right ) 
   \]
   The only satisfying interpretation is $ \{ -1, 2, 3, -4, -5 \}, $ 
   which corresponds to $ \Theta_2. $
   In order to apply the second translation, assume that
   \[ [X/0] = 6, \ \ 
      [X/1] = 7, \ \ 
      [Y/0] = 8, \ \
      [Y/1] = 9, \ \ 
      [Z/0] = 10, \ \ 
      [Z/1] = 11. \] 
   The second translation constructs 
   \[
       {\rm Part}~1:
       \left (
       \begin{array}{ccc}
          1 & 2 \\
          3 & 4 & 5 \\
       \end{array}
       \right )
       \ \ {\rm Part}~2:
       \left (
       \begin{array}{ccc}
          -1 & 6 \\
          -1 & 9 \\
          -2 & 7 \\
          -2 & 8 \\
          -3 & 8 \\
          -3 & 10 \\
          -4 & 8 \\
          -4 & 11 \\
          -5 & 9 \\
          -5 & 10 \\
       \end{array}
       \right )
       \ \ 
       {\rm Part}~3:
       \left (
       \begin{array}{ccc}
          -6 & -7 \\
          -8 & -9 \\
          -10 & -11 \\  
      \end{array}
      \right )
      \ \ {\rm Part}~4:
      \left (
      \begin{array}{ccc}
         -6 & -10 \\
         -7 & -11 \\
      \end{array}
      \right ) 
   \]
   Its only satisfying interpretation is
   $ \{ -1, 2, 3, -4, -5, -6, 7, 8, -9, 10, -11 \}, $ which
   again corresponds to $ \Theta_2. $ 
\end{exa}

\section{An Input Format for GCSP} 
\label{Sect_input_format}

\noindent
Since we are claiming that GCSPs are fundamental enough
to study on their own, and may have applications outside
of geometric resolution, we made our implementation publicly 
available (\cite{deNivelle2018a}).
In this section, we define the input format for GCSP,
which is used by our implementation.  
The format is similar to the DIMACS format
for satisfiability (\cite{DimacsSAT1993}).
Similar to the DIMACS format, variables and constants
are represented by integers. 
Since GCSP has no polarity (there is no negation), all integers
are non-negative.

\begin{defi}
   \label{Def_DIMACS}
   We define a representation for GCSP. 
   Input is represented in plain ASCII. The format never distinguishes
   between upper and lower case. 
   \begin{itemize}
   \item
      Input starts with whitespace, possibly mixed with
      comment lines.
      A comment line is a line whose first non-whitespace character is
      a \verb+`c'+ or a \verb+`C'+. The initial comment lines are ignored.

   \item
      After that comes a line of form
      \begin{verbatim}
      p gcsp nrvars nrconsts nrclauses nrblockings \end{verbatim}
      \verb+nrvars+ is the number of variables in the problem,
      \verb+nrconsts+ is the number of constants in the problem.
      Both need not be exact, but must be upperbounds. More
      precisely, both variables and constants are represented by
      non-negative integers, and \verb+nrvars,nrconsts+ must be bigger
      than any variable or constant that appears in the problem.
 
      \verb+nrclauses+ must be the exact number of clauses,
      and \verb+nrblockings+ must be the exact number of blockings.
 
   \item
      A clause has form \verb+V var1 ... varV S subst1 ... substS+.
      Here \verb+V+ is the exact number of variables in the clause,
      and \verb+var1 ... varV+ are the variables, represented
      by non-negative integers. Each variable must be less than
      \verb+nrvars+. 
      
      \verb+S+ is the exact number of substlets in the clause.
      Each substlet is represented by a sequence of non-negative
      integers of length \verb+V+, that specifies the values
      assigned to the variables, in the same order as the variables.
      Each value must be less than \verb+nrconsts+.
      There must be exactly \verb+nrclauses+ clauses. 
   \item
      Blockings are represented in the same way as clauses.
      Although in Definition~\ref{Def_gcsp}, blockings are single substlets,
      it is convenient to merge blockings with identical domain 
      into clauses, so that they can be represented more 
      compactly.
   
      There must be exactly \verb+nrblockings+ (merged) blockings.
      Note that it is not obligatory to merge blockings with
      same domain. 
   \item
      Everything after the blockings is ignored, so there is room
      for more comments.
   \item
      Solutions are presented in the format
      \begin{verbatim}
         A V1 C1 ... VA CA \end{verbatim}
      Here \verb+A+ is the number of assignments in the substitution,
      and each \verb+Vi+ $ \Rightarrow $ \verb+Ci+ is an assignment.
      The assigments can be listed in arbitrary order. 
   \end{itemize}
\end{defi}

\begin{exa}
   We represent the GCSP of Example~\ref{Example_SAT1_SAT2}.
   There are three variables $ X,Y,Z, $ which we will
   represent by $ 0,1,2. $ This means that $ 3 $ is an upperbound.
   There are two constants $ 0,1, $ so that $ 2 $ is an upperbound.
   This is a representation:  
   \begin{verbatim}

c we did not merge the blockings

p gcsp 3 2 2 2  

2  0 1   2  0 1  1 0 
2  1 2   3  0 0  0 1  1 0
      
2  0 2   1  0 0 
2  0 2   1  1 1  
   \end{verbatim}

   Since the two blockings in Example~\ref{Example_SAT1_SAT2} have the same
   domain, the GCSP can be alternatively represented as follows: 
   \begin{verbatim}    
c  this time we merged the two blockings
c  there is no difference in meaning 

P GCSP 3 2 2 1  

2  0 1   2  0 1  1 0   
2  1 2   3  0 0  0 1  1 0
      
2  0 2   2  0 0  1 1 
   \end{verbatim}

   \noindent
   The solution $ \Theta_2 $ can be output as
   \verb+ 3  0 1  1 0  2 0+. 
   Since the order is arbitrary, it can also be
   output as \verb+ 3  1 0  2 0  0 1+. 
\end{exa}

\section{Experiments}
\label{Sect_experiments}

\noindent
We present measurements on two benchmark sets. Both sets were
obtained by running {\bf Geo} on a few input problems, and collecting
hard matching instances. The first set consists of problems 
that took more than an hour to solve with a naive matching algorithm.
This set was used in Figures~\ref{Fig_unary_sat} and 
\ref{Fig_tyebye_z_maslem}. 

\begin{figure}
   \caption{Comparing Direct Matching with SAT Translations} 
   \label{Fig_unary_sat}
\[
   \begin{array}{|l||c|c||c|c|}
      \hline
      {\bf Problem} & {\rm Algo}~\ref{Algo_unary} & {\rm Algo}~\ref{Algo_unary} ({\bf flat})& 
                      {\rm Def}~\ref{Def_trans_SAT1} & {\rm Def}~\ref{Def_trans_SAT2} \\
      \hline
      {\bf mod01}       & 271(42606)  & 0.15(2898)   &  0.17 (8248)     & 0.08(3220)   \\
      {\bf mod02}       & 138(28830)  & 0.11(2064)   &  0.19 (8248)     & 0.12(5982)   \\
      {\bf mod03}       & 80(21822)   & 0.095(1632)  &  0.6 (20344)     & 0.17(5288)   \\
      {\bf mod04}       & 32(14290)   & 0.062(1440)  &  0.05 (2955)     & 0.028(1744)  \\ 
      {\bf mod05}       & 21(11640)   & 0.049(1210)  &  0.06 (5007)     & 0.036(1747)  \\
      {\bf mod06}^{*}   & 340(23193)  & 1.42(9862)   &  0.06 (1637)     & 0.044(1637)  \\
      {\bf mod07}^{*}   & 703(31347)  & 2.05(12495)  &  0.14 (5032)     & 0.098(2955)  \\
      {\bf mod08}^{*}   & 1593(42709) & 3.06(17947)  &  0.87 (20658)    & 0.15(5032)   \\
      \hline
      {\bf mod22}^{*}   & 133(25620)  & 2.8(18542)   &  10.14 (113548)  & 100(110445)  \\
      {\bf mod23}^{*}   & 52(17533)   & 2.09(15114)  &  75.64 (230822)  & 8.9(58213)   \\
      {\bf subst15}^{*} & 0.38(4)     & 0.05(4)      &  0.07 (350)      & 0.060(98)    \\ 
      {\bf syn02}^{*}   & 0.0017(0)   & 0.00014(0)   &  0.14 (0)        & 0.0035(0)    \\
      {\bf syn11}^{*}   & 0.0006(4)   & 0.00017(2)   &  0.024 (0)       & 0.0098(0)    \\
      {\bf syn12}^{*}   & 0.0022(1)   & 0.00025(0)   &  0.13 (0)        & 0.012(0)     \\
      {\bf syn14}^{*}   & 3.81(1461)  & 0.31(1083)   &  0.04 (0)        & 0.035(1634)  \\
      \hline
   \end{array}
\]
\end{figure}

\noindent
Entries in Figure~\ref{Fig_unary_sat} have form $ t(\lambda), $ where $ t $ is the time used in
seconds, and $ \lambda $ the number of lemmas generated. 
For the 3d and 4th column, the times are the CPU-times reported by MiniSat 
(Version~2.0~beta) (\cite{Minisat2004}).
Since MiniSat is not integrated into 
{\bf Geo}, it is difficult to measure the total time (conversion+solving).
For hard problems, the conversion times are probably negligible, but for
trivial problems, they may be significant (in the same order
of magnitude as the solving times) because translation is quadratic. 
Due to the way the benchmarks were collected, they contain
no trivial problems.
It can be seen from Figure~\ref{Fig_unary_sat} that translation to 
propositional SAT is comparable to Algorithm~\ref{Algo_unary} with
flat lemmas. We will dicuss this more in the context of 
Figure~\ref{Fig_unary_sat2}.  
We were not sure how to determine the number of lemmas generated
during a run of MiniSat, due to the fact that it performs restarts.
Currently, we simply added the numbers reported by the different 
restarts. Since MiniSat probably reuses lemmas between different 
restarts, this means that the indicated numbers are likely too high. 
It can be seen that nearly always, 
Definition~\ref{Def_trans_SAT2} performs better than 
Definition~\ref{Def_trans_SAT1}. 

\begin{figure}
   \caption{Results for Matching Using Local Consistency} 
   \label{Fig_tyebye_z_maslem}
\[
   \begin{array}{|l||c|c|c||c|}
      \hline
      {\bf Problem} & S=1 & S=2 & 
                      S=3 & S=1 \ ({\bf flat}) \\
      \hline
      {\bf mod01}       & 257(104268)  & 256(104268)   & 283(104268)  & 1712(670938)  \\
      {\bf mod02}       & 334(90012)   & 359(90012)    & 330(90012)   & 1397(585402)  \\
      {\bf mod03}       & 148(75288)   & 142(75288)    & 150(75324)   & 639(464658)   \\
      {\bf mod04}       & 42(35985)    & 41(35985)     & 46(35985)    & 182(194905)   \\
      {\bf mod05}       & 39(35110)    & 40(35110)     & 43(35110)    & 133(190530)   \\
      {\bf mod06}^{*}   & 577(27689)   & 602(27689)    & 443(32467)   & 593(125993)   \\
      {\bf mod07}^{*}   & 946(32338)   & 962(32338)    & 669(35410)   & 888(155391)   \\
      {\bf mod08}^{*}   & 1580(42193)  & 1719(42193)   & 1091(39057)  & 1258(175607)  \\
      \hline 
      {\bf mod22}^{*}   & 379(30758)  &  355(30758)    & 243(26814)   & 62(53035)     \\
      {\bf mod23}^{*}   & 92(18228)   &  91(18228)     & 69(14662)    & 26(26384)     \\
      {\bf subst15}^{*} & 0.42(43)    &  0.44(43)      & 0.66(43)     & 0.42(43)      \\
      {\bf syn02}^{*}   & 0.013(0)    &  0.01(0)       & 0.015(0)     & 0.012(0)      \\
      {\bf syn11}^{*}   & 0.006(2)    &  0.019(26)     & 0.063(14)    & 0.006(2)      \\
      {\bf syn12}^{*}   & 0.01(0)     &  0.015(0)      & 0.031(0)     & 0.008(0)      \\
      {\bf syn14}^{*}   & 0.098(132)  &  0.12(126)     & 0.33(92)     & 0.0023(195)   \\
      \hline
   \end{array}
\]
\end{figure}

Figure~\ref{Fig_tyebye_z_maslem} shows results for the refining algorithm
of \cite{deNivelle2016a}, discussed in Section~\ref{Sect_IJCAR2016}.
It can be seen that using $ S > 1 $ is hardly worth the effort, and
that the refining algorithm performs somewhat worse than 
Algorithm~\ref{Algo_unary} with learning of unrestricted lemmas. 
Since flattening of lemmas improves the
performance of Algorithm~\ref{Algo_unary} dramatically,
we tried the same with the refining algorithm. Unfortunately,
the last column
of Figure~\ref{Fig_tyebye_z_maslem} shows that flattening has 
a big, negative impact on the refining algorithm. 
That means that the refining algorithm can be ruled out as a candidate
for being optimal. 

\begin{figure}
   \caption{Algorithm~\ref{Algo_unary} (flat lemmas) 
            against {\rm Def}~\ref{Def_trans_SAT2}} 
   \label{Fig_unary_sat2}
\[
   \begin{array}{||c|c||c|c|}
      \hline
      {\rm Algo}~\ref{Algo_unary} ({\bf flat}) & {\rm Def}~\ref{Def_trans_SAT2} &
      {\rm Algo}~\ref{Algo_unary} ({\bf flat}) & {\rm Def}~\ref{Def_trans_SAT2} \\
      \hline
      0.66(10565)  & 5.6(8169)     &  0.16 (2875)     & 7.29(23457)   \\
      0.20(4081)   & 4.02(13441)   &  0.13 (3231)     & 0.95(6551)    \\
      0.18(4521)   & 1.07(10670)   &  0.00013(5)      & 0.0(0)        \\
      0.00019(4)   & 0.0029(0)     &  0.037(1182)     & 0.043(2058)   \\
      0.057(317)   & 0.78(7478)    &  0.058(683)      & 1.89(16973)   \\
      \hline
      0.147(558)   & 2.88(16971)   &  0.053 (1795)    & 0.12(2810)    \\
      0.21(1039)   & 3.08(17065)   &  0.026 (910)     & 0.072(1749)   \\
      0.033(1006)  & 0.032(1318)   &  0.018 (402)     & 0.022(1299)   \\
      0.23(6308)   & 0.073(3210)   &  0.023 (310)     & 0.101(3216)   \\
      0.026(990)   & 0.027(1318)   &  0.0655 (1208)   & 0.11(3215)    \\
      \hline
                   &               &  0.036 (553)     & 0.069(2070)   \\
      \hline
   \end{array}
\]
\end{figure}

\noindent
Figure~\ref{Fig_unary_sat2} presents another benchmark test that 
was obtained by running {\bf Geo} on several input problems using 
Algorithm~\ref{Algo_unary}
with flat lemmas. We started by setting
a short initial time $ t = 10^{-4}. $
Whenever a matching took more than $ t $ seconds to solve, we 
added it to the benchmark set, and doubled $ t. $ 
From several runs, we kept the last three instances generated
in this way. 
Figure~\ref{Fig_unary_sat2} shows that most of the problems obtained 
in this way, are also hard for MiniSat.
We believe Figure~\ref{Fig_unary_sat2} shows the potential of direct 
matching algorithms,
but there are several caveats: 
Runs of {\bf geo} do not generate really hard matching instances,
all instances are solved within seconds, most much faster.
The way of collecting problems that are hard for 
Algorithm~\ref{Algo_unary} puts
it at a disadvantage. For example, there may be problems that
are hard for MiniSat, which will not enter the benchmark set.
On the other hand, MiniSat is not state of the art anymore,
and modern SAT solvers probably will perform better.
We carefully conclude that there is a chance that in the long
term, at least for some applications, our approach of directly
implementing matching, may be the optimal approach.

\section{Finding Optimal Matchings}
\label{Sect_optimal}

\noindent
In this section we address the problem of finding optimal matchings.
For the effectiveness of geometric resolution, it is important
that a minimal matching is returned, in case more than one exists.
A minimal matching is a matching that uses the smallest possible
set of assumptions. In terminology of DPLL, assumptions represent
decision levels. The assumptions contributing to a conflict
represent choice options,
which will be replaced by other options during backtracking.
In addition to being as few as possible,
assumptions at a lower decision level should always be preferred over 
assumptions at
a higher decision level. The reason for this is the fact that
in other branches of the search tree, there is a risk that
more assumptions will be used, and when assumptions are at a lower level, 
there is less room for this.  

\begin{defi}
   Let $ I $ be an interpretation. A weight function $ \alpha $
   is a function that assigns finite subsets of natural numbers to
   the atoms of $ I. $ 

   Let $ A $ be a geometric literal. Let $ \Theta $ be a substitution
   such that $ A \Theta $ is in conflict with $ I. $ 
   Referring to definition~\ref{Def_conflict_truth}, we
   define $ \alpha( \ p_{\lambda}( x_1, \ldots, x_n) \Theta, I ) = 
              \alpha( \ p_{\mu}( x_1 \Theta, \ldots, x_n \Theta ) \ ), \ \ 
         \alpha( \ (x_1 \approx x_2) \Theta, I ) = \{ \}, $ and
   $ \alpha( \ ( \#_{\bf f} x ) \Theta, I ) = 
             \alpha( \ ( \#_{\bf t} x \Theta ) \ ). $ 
\end{defi}

\begin{defi}
   \label{Def_matching_weight}
   Let $ I $ and 
   $ \phi = \ A_1, \ldots, A_p \ | \ B_1, \ldots, B_q $ together form
   an instance of the matching problem (Definition~\ref{Def_matching}).
   Assume that $ \Theta $ is a solution.  
   The \emph{weight of} $ \Theta, $ for which we write 
   $ \alpha( I, \phi, \Theta ), $ is defined as 
   \[ \bigcup \left \{
      \begin{array}{l}
         \{ \ \alpha( A_i \Theta, I ) \ | \ 1 \leq i \leq p \} \\
         \{ \ \alpha( C, I ) \ | \ 
             1 \leq j \leq q, \ \ C \in E( B_j, \Theta ), \mbox{ and }
             C \mbox{ conflicts } I \ \} \\
      \end{array}
      \right .
   \] 
   Solving optimal matching means: 
   First establish if $ (I, \phi ) $ has a solution. 
   If it has, then find a solution $ \Theta $ for which
   $ \alpha(I, \phi, \Theta ) $ is multiset minimal.
\end{defi}
One could try to impose further selection criteria that are harder
to explain and whose advantage is less evident.

Solving the minimal matching problem is non-trivial, because the number
of possible solutions can be very large. 
The straightforward solution is to use some efficient
algorithm (e.g. the one in this paper)
that enumerates all solutions, and keeps the best solution.
Unfortunately, this approach is completely impractical because 
some instances have a very high number of solutions.
One frequently encounters 
instances with $ > 10^{9} $ solutions.

In order to find a minimal solution without enumerating
all solutions, one can use any algorithm that stops on the first
solution in the following way:  
The first call is used to find out whether a solution 
exists. If not, then we are done.
Otherwise, the algorithm is called again with its input
restricted in such a way 
that it has to find a better solution than the
previous.
One can continue doing this, until all possibilities to improve
the solution have been exhausted. 
It can be shown that the number of calls needed to obtain an optimal
solution is linear in the size of the assumption set of solution. In this way,
it can be avoided that all solutions have to be enumerated.

\begin{defi}
   \label{Def_restricted_translation}
   Let $ I $ be an interpretation that is equipped 
   with a weight function $ \alpha. $ 
   Let $ \phi = \ A_1, \ldots, A_p \ | \ B_1, \ldots, B_q $ be a
   geometric formula. 
   Let $ \alpha $ be a fixed set of natural numbers. 
   We define the $ \alpha $-restricted translation 
   $ ( \Sigma^{+}, \Sigma^{-} ) $ of 
   $ ( I, \phi ) $ as follows: 

   \begin{itemize}
   \item
      For every $ A_i, $ let $ \overline{v}_i $ be the variables
      of $ A_i. $ 
      Then $ \Sigma^{+} $ contains the clause  
      \[ \{ \overline{v}_i / \overline{v}_i \Theta \ | \ 
                A_i \Theta \mbox{ is in conflict with } I \mbox{ and } 
                \alpha( A_i \Theta, I ) \subseteq \alpha \}. \]
   \item
      For each $ B_j, $ let $ \overline{w}_j $ denote the variables of
      $ B_j. $ For every $ \Theta $ that makes $ B_j \Theta $ true
      in $ I, $ \ \ 
      $ \Sigma^{-} $ contains the substlet
      $ \overline{w}_j / ( \overline{w}_j \Theta ). $ 
      In addition, if there exists a $ C \in E( B_j, \Theta ) $ 
      that is in conflict with $ I $ and for which 
      $ \alpha( C, \Theta ) \not \subseteq \alpha, $ then
      $ \Sigma^{-} $ contains the substlet 
      $ \overline{w}_j / ( \overline{w}_j \Theta ). $ 
   \end{itemize}

\end{defi}
   
\noindent
The $ \alpha $-restricted translation ensures that only conflicts
involving atoms $ C $ with $ \alpha(C) \subseteq \alpha $ are considered,
and (independently of $ \alpha $), that no $ B_j $ is made true. 
The translation of Definition~\ref{Def_trans_gcsp} can be viewed as a special
case of $ \alpha $-restricted translation with $ \alpha = \mathbb{N}. $

\begin{thm}
   \label{Thm_alpha_restricted}
   Let $ ( \Sigma^{+}, \Sigma^{-} ) $ be obtained by $ \alpha $-restricted
   translation of $ (I, \phi). $ 
   For every substitution $ \Theta, $ \ \ 
   $ \Theta $ is a solution of $ ( \Sigma^{+}, \Sigma^{-} ) $ 
   iff 
   $ \Theta $ is a solution of $ ( I, \phi ), $ and it 
   has $ \alpha( I, \phi, \Theta ) \subseteq \alpha. $ 
\end{thm}

\noindent
Using $ \alpha $-restricted translation, we can define 
the {\bf optimal} matching algorithm:

\begin{algo}
   \label{Alg_find_optimal_match}
   Let $ {\bf solve}( \Sigma^{+}, \Sigma^{-} ) $ be a function that returns
   some solution of $ ( \Sigma^{+}, \Sigma^{-}) $ if it has a solution, 
   and $ \bot $ otherwise.

   We define the algorithm
   $ {\bf optimal}( \ I, \phi \ ) $ 
   that returns an optimal solution of $ (I,\phi) $ if one exists and 
   $ \bot $ otherwise. 
   \begin{enumerate}
   \item
      Let $ ( \Sigma^{+}, \Sigma^{-} ) $ be the GCSP obtained by
      the translation of Definition~\ref{Def_trans_gcsp}. 
      If $ \Sigma^{+} $ contains an empty clause, then return $ \bot. $ 
      If $ \Sigma^{-} $ contains a propositional blocking, then return
      $ \bot. $ 
      Otherwise, remove unit blockings from $ ( \Sigma^{+}, \Sigma^{-} ). $ 
      If this results in $ \Sigma^{+} $ containing an empty clause, then 
      return $ \bot. $ 
   \item
      Let $ \Theta = {\bf solve}( \Sigma^{+}, \Sigma^{-} ). $ 
      If $ \Theta = \bot, $ then {\bf return} $ \bot. $ 

   \item
      Let $ \alpha = \alpha(I, \phi, \Theta), $ and
      let $ k := \sup( \alpha ). $ 
   
   \item
      As long as $ k \not = 0, $ do the following:
      \begin{itemize} 
      \item
         Set $ k = k - 1. $ 
         If $ k \in \alpha, $ then do 
         \begin{itemize}
         \item
            Let $ \alpha' = 
                    ( \alpha \backslash \{ k \} ) \cup 
                \{ 0, 1, 2, \ldots, k - 1 \}. $ 
         \item
            Let $ ( \Sigma^{+}, 
                    \Sigma^{-} ) $ be the
            $ \alpha' $-restricted translation 
            of $ (I, \phi). $ 
         \item
            If $ \Sigma^{+} $ contains an empty clause
            or $ \Sigma^{-} $ contains a propositional blocking, 
            then skip the rest of the loop.
            Otherwise, remove the unit blockings from 
            $ ( \Sigma^{+}, \Sigma^{-} ). $ 
            If this results in $ \Sigma^{+} $ containing the empty clause, 
            then skip the rest of the loop.
         \item
            Let $ \Theta' = {\bf solve}( \Sigma^{+}, \Sigma^{-} ). $ 
            If $ \Theta' \not = \bot, $ then set $ \Theta = \Theta' $
            and $ \alpha = \alpha( I, \phi, \Theta ). $ 
         \end{itemize} 
      \end{itemize}    
   \item
      Now $ \Theta $ is an optimal solution,
      so we can {\bf return} $ \Theta. $ 
   \end{enumerate}
\end{algo}
Algorithm $ {\bf optimal} $ first solves $ (I,\phi) $ without
restriction. If this results in a solution $ \Theta, $ it checks for each 
$ k \in \alpha( I, \phi, \Theta ) $ if $ k $ 
can be removed. The invariant of the main loop is:
There exists no $ k' \geq k $ that occurs in $ \alpha( I, \phi, \Theta ), $ 
and no $ \Theta' $ that is a solution of $ (I, \phi) $ with
$ k' \not \in \alpha( I, \phi, \Theta' ). $ 
In addition, the invariant $ \alpha = \alpha( I,\phi,\Theta ) $ 
is maintained.

\begin{exa}
   Assume that in example~\ref{Ex_matchings}, the
   atoms have weights as follows:
   \[ \begin{array}{ll}
         \alpha( \ P_{\bf t}( c_0, c_0 ) \ ) = \{ 1 \}, \ \
         \alpha( \ P_{\bf e}( c_0, c_1 ) \ ) = \{ 2 \}, \ \
         \alpha( \ P_{\bf t}( c_1, c_1 ) \ ) = \{ 3 \}, \\
         \alpha( \ P_{\bf e}( c_1, c_2 ) \ ) = \{ 4 \}, \ \
         \alpha( \ Q_{\bf t}( c_2, c_0 ) \ ) = \{ 5 \}. \\
      \end{array}
   \]
   We have $ \alpha( I, \phi_1, \Theta_1 ) = \{ 1 \}, $ \ \ 
   $ \alpha( I, \phi_1, \Theta_2 ) = \{ 1,2 \}, $ 
   and $ \alpha( I, \phi_1, \Theta_3 ) = \{ 2,3 \}. $ 
   If $ \Theta_3 $ is the first solution generated,
   $ {\bf solve} $ will construct the 
   $ \{ 1,2 \} $-restricted translation of $ (I,\phi_1), $ 
   which equals 
   \[ \begin{array}{l}
         (X,Y) \ / \ ( c_0,c_0 ) \ | \ ( c_0,c_1 ) \\
         (Y,Z) \ / \ ( c_0,c_0 ) \ | \ ( c_0,c_1 ) \\
         \hline
         (X,Z) \ / \ ( c_0,c_2 ) \\
      \end{array}
   \]
   If the next solution found is $ \Theta_2, $ then
   $ {\bf solve} $ will construct the $ \{ 1 \} $-restricted translation
   \[ \begin{array}{l}
         (X,Y) \ / \ ( c_0,c_0 ) \\
         (Y,Z) \ / \ ( c_0,c_0 ) \\
         \hline
         (X,Z) \ / \ ( c_0,c_2 ) \\
      \end{array}
   \]
   whose only solution is $ \Theta_1. $ 
\end{exa}

\section{Filtering by Local Consistency Checking} 
\label{Sect_local_consistency_checking}

\noindent
Filtering is any procedure that simplifies or possibly rejects a
GCSP before the main algorithm is called.
In {\bf Geo}, we have used filtering based on local consistency
checking. In earlier versions, this was effective because
very often, filtering rejects a GCSP without calling the main algorithm.
Since the algorithms that we present in this paper, are much
more efficient, this is not certain anymore. 
We still present the local consistency checking procedure,
because it is easy to implement using refinement stacks,
and it may be still an effective tool for filtering out
easy instances. 

Local consistency checking 
(see \cite{Dechter2003,MaloSebag2004,SchefHerbWys94})
is a pre-check that comes in many variations. 
Local consistency checking is the following procedure:
For every clause $ c = \{ s_1, \ldots, s_n \} \in \Sigma^{+}, $ check,
for all sets of clauses $ C $ with size $ S \geq 1, $
if $ \{ s_i \} \cup C $ has a solution. If not, then
remove $ s_i $ from $ c. $
Keep on doing this, until no further changes are possible or a clause
has become empty. Local consistency checking rejects a large
percentage of GCSP instances a priori, and usually decreases
the size of the clauses involved by a factor two or three.

In \cite{Dechter2003} (Chapter~3), local consistency checking is
defined using subsets of variables (instead
of clauses). Using subsets of two variables is called
\emph{arc consistency checking}, while considering subsets of 
three variables is called \emph{path consistency checking}.
In general, using bigger subsets is a more effective precheck,
but also more costly because it gets closer to the original problem. 

As discussed in Section~\ref{Sect_IJCAR2016}, we had assumed in 
\cite{deNivelle2016a} 
that filtering is so effective, that one can base the complete search
algorithm on it. Although this is possible in theory, the resulting
algorithm turned out not competitive.

Since the local consistency checks the substlets in a single
clause $ c $ against sets of clauses $ C \subseteq \Sigma^{+}, $
we define the size $ S $ of a local consistency check as $ S = \| C \|. $ 
When performing a local consistency check up to size $ S, $ 
one has to generate subsets up to size $ S+1, $ and generate
their solutions. If $ \| \Sigma^{+} \| = n, $ the total
number of such subsets equals 
$ \left ( \begin{array}{l} n \\ S+1 \\ \end{array} \right ), $ which 
grows very quickly for realistic $ n. $ 
The problem can be decreased by not generating all subsets,
but only generate subsets whose clauses 
share variables, or have variables that co-occur in a blocking.
\begin{defi}
   \label{Def_related}
   Let $ c,c' $ be clauses. We write $ c \sim c' $ if either
   $ c $ and $ c' $ share a variable, or there exist connected
   (Definition~\ref{Def_connected_variable}) variables $ v $ and $ v', $
   s.t. $ v $ occurs in $ c $ and $ v' $ occurs in~$c'. $
\end{defi}

\noindent
It is sufficient to generate subsets that are connected,
because consideration of subsets that are not connected will not lead to 
the removal of more substlets. We always assume that solutions
are non-redundant, i.e. do not contain irrelevant assignments.

\begin{lem}
   \label{Lemma_disconnected}
   Let $ ( \Sigma^{+}, \Sigma^{-} ) $ be a GCSP.
   Let $ C \subseteq \Sigma^{+}. $ If $ C $ can be written as
   $ C_1 \cup C_2, $ s.t. there exist no $ c_1 \in C_1 $ 
   and no $ c_2 \in C_2 $ with $ c_1 \sim c_2, $ then 
   for every two substitutions $ \Theta_1, \Theta_2, $ s.t. 
   $ \Theta_1 $ is a solution of $ ( C_1, \Sigma^{-} ) $ and
   $ \Theta_2 $ is a solution of $ ( C_2, \Sigma^{-} ), $ 
   \ \ $ \Theta_1 \cup \Theta_2 $ is a solution of 
   $ ( C_1 \cup C_2, \Sigma^{-} ). $
\end{lem}
Lemma~\ref{Lemma_disconnected} guarantees that it is not needed
to attempt to remove substlets from clauses in $ C_1 \cup C_2, $ 
after $ C_1 $ and $ C_2 $ have been checked. 
If some substlet $ s $ in $ C_1 $ occur some solution of $ C_1, $ and
$ C_2 $ has a solution, then $ s $ will occur in the combined solution.

We will now show that instead of ignoring disconnected subsets,
one can also ignore subsets that are connected only through a single
clause:

\begin{lem}
   \label{Lemma_pivot} 
   Let $ ( \Sigma^{+}, \Sigma^{-} ) $ be a GCSP. 
   Assume that $ C_1, C_2 \subseteq \Sigma^{+}, $ and 
   $ c \in \Sigma^{+}. $  
   Assume that for every pair of variables 
   $ v_1 $ occurring in a clause of $ C_1, $ and
   $ v_2 $ occurring in a clause of $ C_2, $ if
   either $ v_1 = v_2 $ or $ v_1 $ and $ v_2 $ are connected,
   then $ v_1, v_2 $ occur in $ c. $ 

   Then the following holds: 
   If $ \Theta_1 $ is a solution of 
   $ ( C_1 \cup \{ c \}, \Sigma^{-} ) $
   and $ \Theta_2 $ is a solution of 
   $ ( C_2 \cup \{ c \}, \Sigma^{-} ), $
   s.t. $ \Theta_1, \Theta_2 $ agree on the variables occurring in $ c, $
   then $ \Theta_1 \cup \Theta_2 $ is a solution 
   of $ ( C_1 \cup C_2 \cup \{ c \}, \Sigma^{-} ). $ 
\end{lem} 
\begin{proof}
   Assume that $ \Theta_1, \Theta_2, C_1,C_2,c $ fulfil the
   conditions of the lemma. 
   By non-redundancy, $ \Theta_1 $ does not contain assignments to
   variables not occurring in $ c $ or $ C_1. $ 
   Similarly, $ \Theta_2 $ does not contain assignments to
   variables not occurring in $ c $ or $ C_2. $  
   If $ \Theta_1, \Theta_2 $ share a variable $ v, $ then
   this variable must occur in $ c, $ 
   which implies that $ v \Theta_1 = v \Theta_2. $ 
   As a consequence, $ \Theta_1 $ and $ \Theta_2 $ can be merged into
   a single substitution $ \Theta = \Theta_1 \cup \Theta_2, $ 
   which has $ \Theta \models C_1 \cup C_2 \cup \{ c \}. $ 
   
   If there would be a blocking $ \sigma \in \Sigma^{-}, $ s.t.
   $ \Theta \models \sigma, $ then we still have
   $ \Theta_1 \not \models \sigma $ and $ \Theta_2 \not \models \sigma. $
   This implies that there are variables $ v_1 $ in 
   $ C_1 \backslash \{ c \} $ and 
   $ v_2 $ occurring in $ C_1 \backslash \{ c \}, $ which 
   occur together in $ \sigma. $ But this 
   contradicts the fact that $ v_1 $ and $ v_2 $ cannot be
   connected. 
\end{proof}

As above, if some substlet $ s \in C_1 $ 
is used in a solution of $ C_1 \cup \{ c \}, $ and $ \{ c \} \cup 
C_2 $ has a solution, then the solutions can be
combined into a single solution that uses $ s. $ 
If some substlet $ s $ of $ c $ occurs in a solution
of $ C_1 \cup \{ c \} $ and in a solution of $ \{ c \} \cup C_2, $ 
then the solutions can be combined into a single solution that
still uses $ s. $ 

This implies that, if one uses a local consistency checker that 
gives preference to small subsets, one can ignore subsets
that do not contain 'cycles'. If there exist $ c_1,c_2,c_3 \in C, $ 
s.t. $ c_1,c_3 $ are not connected, and every path from $ c_1 $ to
$ c_3 $ has to pass through $ c_2, $ then $ C $ can be ignored.
This gives rise to the following definition: 

\begin{defi}
   \label{Def_circle}
   Let $ ( c_1, \ldots, c_{S+1} ) $ with $ S \geq 1 $ be a sequence 
   of clauses. We call $ ( c_1, \ldots, c_{S+1} ) $ \emph{a circle} if 
   for every $ i \ ( i \leq S ), $ we have $ c_i \sim c_{i+1}, $ 
   and in addition we have $ c_{S+1} \sim c_1. $ 

   If $ \overline{C} $ is a refinement stack, we call a sequence
   of indices $ ( i_1, \ldots, i_{S+1} ) $ \emph{a circle} if
   each $ \alpha_{i_j}( \overline{C} ) $ is true, and
   $ ( d_{i_1}, \ldots, d_{i_{S+1}} ) $ is a circle.
\end{defi}

The local consistency checker checks only circles. 
Generation of circles
in $ \Sigma^{+} $ is easier to implement than generation of all connected 
subsets, especially if one wants to avoid generating the same subset in
different ways. In addition, it is more efficient because there are less
circles than connected subsets. 
The discussion above suggests that generating circles
is sufficient to obtain a complete check. We have believed for 
some time that this is true in general, but we will show below that 
it is false. 
 
\begin{algo}
   \label{Algo_local_consistency_checking}
   Let $ S \geq 1 $ be a natural number.
   Let $ \Theta $ be a substitution.
   Let $ \overline{C} $ be a refinement stack. 

   A call to $ {\bf local}( s, \Theta, 
                            ( k_1, \ldots, k_{S+1} ), \overline{C} ) $ 
   constructs a refinement of $ \overline{C} $ by removing the substlets
   that do not occur in any solution of the subset of $ \Sigma^{+} $
   of size $ S + 1. $ 
   
   It returns $ \bot $ if it establishes that $ \Theta $ cannot
   be extended into a solution of $ \overline{C}. $ 
   Initially $ s = k_1 = \cdots = k_{S+1} = 1. $ 
   \begin{description}
   \item[SUBST]
      As long as $ s \leq \| \Theta \|, $ let $ v/x $ be the
      $ s $-th assignment in $ \Theta. $ 
      \begin{enumerate}
      \item
         For every blocking $ \sigma \in \Sigma^{-} $ involving
         $ v, $ check if
         $ \Theta \models \sigma. $ If yes, then return $ \bot. $ 
      \item
         For every $ ( c_i \Rightarrow d_i ) \in \overline{C} $ which
         has $ \alpha_i( \overline{C} ) $ true and which contains 
         $ v, $ let $ d' $ be the set of substlets in $ d_i $ that
         are consistent with $ \Theta. $ 
         If $ d' = \emptyset, $ then return $ \bot. $ 
         Otherwise, if $ \emptyset \subset d' \subset d, $ append
         $ ( c_i \Rightarrow d' ) $ to $ \overline{C}. $ 
      \end{enumerate} 
   \item[CLAUSES1]
      As long as $ k_1 < \| \overline{C} \| $ do the following:
      \begin{enumerate}
      \item
         If $ \alpha_{k_1}( \overline{C} ) $ is true, and
         the $ k_1 $-refinement $ ( c_{k_1} \Rightarrow d_{k_1} ) $ 
         contains a variable $ v,$
         s.t. all substlets $ (\overline{v}/\overline{x}) \in d_{k_1} $ 
         agree on the assignment to $ v, $ then let $ x $ be the agreed 
         value. Append $ v/x $ to $ \Theta. $ 
      \item
         Set $ k_1 := k_1 + 1. $  
      \end{enumerate}
      If $ s \leq \| \Theta \|, $ then restart at {\bf SUBST}. 
      (This means that $ \Theta $ was extended in the previous step.) 
   \item[CLAUSESN] 
      As long as there is an $ i $ with $ 2 \leq i \leq S+1, $ 
      s.t. $ k_i \leq \| \overline{C} \|, $ 
      pick the smallest such $ i. $ 
      If $ \alpha_{k_i}( \overline{C} ) $ holds, then 
      \begin{enumerate}
      \item
         Enumerate all circles $ ( \lambda_1, \ldots, \lambda_{i} ) $
         of size $ i $ starting at $ \lambda_1 = k_i. $ 
         For each such circle $ ( \lambda_1, \ldots, \lambda_{i} ), $ 
         let $ I = \{ d_{\lambda_1}, \ldots, d_{\lambda_{i}} \}. $ 

         Call $ {\bf refine}(I, \Theta, \overline{C} ). $ 
         If the result is $ \bot, $ then return $ \bot. $ 
         If after the call, we have 
         $ \| \overline{C} \| > k_1, $ then restart at
         {\bf CLAUSES1}.   
      \end{enumerate} 
   \end{description}

   \noindent
   Let $ \overline{C} $ be a refinement stack. Let
   $ k = \| \overline{C} \|. $ 
   Let $ I $ be a subset of $ \{ 1, \ldots, k \}, $ s.t.
   for every $ i \in I, \ \ \alpha_i( \overline{C} ) $ holds.
   Algorithm 
   $ {\bf refine}( I, \Theta, \overline{C} ) $ is defined as follows:  

   \begin{enumerate}
   \item
      Initialize a map $ U $ with domain $ I $ 
      by setting $ U(i) = \emptyset, $ for each
      $ i \in I. $ 
      Eventually, $ U $ will map 
      each $ i \in I $ to the set of substlets in $ d_i, $ that can occur 
      in a solution $ \Theta' $ of 
      $ \{ d_i \ | \ i \in I \} $ extending $ \Theta. $ 
   \item
      Enumerate all maps $ S $ with domain $ I $ 
      that map each $ i \in I $ to 
      a substlet $ S(i) $ in $ d_i, $ and that have the following properties: 
      No $ S(i) $ conflicts $ \Theta, $ 
      no $ S(i), S(i') $ are in conflict with each other. 
      $ \Theta \cup \{ S(i) \ | \ i \in I \} $ does not imply a blocking
      $ \sigma \in \Sigma^{-}. $  

      \noindent
      For each of the generated mappings $ S, $ for each $ i \in I, $ set
      $ U(i) = U(i) \cup \{ S(i) \}. $ 
   \item
      For every $ i \in I, $ for which $ U(i) \not = d_i, $ 
      add the refinement $ ( \ c_i \Rightarrow U(i) \ ) $ to $ \overline{C}. $ 
   \end{enumerate} 
\end{algo}

The local consistency checker gives priority to 
checking against the substitution. 
After checking for conflicts against the substitution, 
Algorithm~\ref{Algo_local_consistency_checking} generates
circles of size up to $ S+1, $ and checks for each of the substlets
occurring in the clauses of such a circle, whether it can occur
in a solution. Substlets that do not occur in a solution are
refined away.
Preference is given to small circles. This means that circles
of size $ i+1 $ will be checked only after all circles
up to size $ i $ have been checked. 

We will discuss (and disprove) the conjecture mentioned above, that 
it is sufficient to check circles, when preference is given
to smaller subsets. More precisely: 
If for a given subset $ C \subseteq \Sigma^{+}, $ all its
subcircles have been checked, then $ C $ needs to be checked
only if it is a circle by itself.
We formally define what 'has been checked' means:

\begin{defi}
   \label{Def_filtered} 
   Let $ ( \Sigma^{+}, \Sigma^{-} ) $ be a GCSP.
   Let $ \Theta $ be a substitution. 
   Let $ C $ be a subset of clauses of $ \Sigma^{+}. $ 
   We write $ \Phi(C) $ for the following property:
   For every clause $ c \in C, $ for every substlet $ s \in c, $ 
   there is a solution $ \Theta $ of $ ( C, \Sigma^{-} ), $
   s.t. $ \Theta \models s. $ 
\end{defi}

Algorithm~\ref{Algo_local_consistency_checking} tries to 
establish $ \Phi(C) $ for every
subset $ C $ of size $ i $ up to $ S+1. $ It assumes that 
when $ \Phi(C) $ holds for circles with size smaller than $ \| C \|, $ 
and $C$ is not a circle, then $ \Phi(C) $ automatically holds.
We have believed for some time that this assumption is true,
because Lemma~\ref{Lemma_disconnected} and 
Lemma~\ref{Lemma_pivot} provide evidence for it,
and it simplifies Algorithm~\ref{Algo_local_consistency_checking}.
Unfortunately, the property fails at $ S = 4, $
when circles have size $ 5. $ 

\begin{conj}
   \label{Conj_circles}
   Let $ ( \Sigma^{+}, \Sigma^{-} ) $ be a GCSP. 
   Assume that every 
   strict subset $ C' \subset \{ c_1, \ldots, c_{S+1} \} $ that
   can be arranged into a circle $ c'_1, \ldots, c'_{S'+1} $
   has property $ \Phi(C'). $ 
   Then if $ C $ cannot be arranged into a circle, $ C $ 
   has the property $ \Phi(C). $ 
\end{conj}
We prove Conjecture~\ref{Conj_circles} for $ S < 4, $ and 
provide a counter example for $ S = 4. $ 
\begin{proof}
   \begin{itemize}
   \item
      $ S = 1 $ follows from Lemma~\ref{Lemma_disconnected}. 
   \item
      In order to prove $ S = 2, $ assume that
      $ c_1, c_2, c_3 $ are clauses that do not form a circle. 
      Without loss of generality, we may assume that
      $ c_1 \not \sim c_3. $ 
      If we also have $ c_1 \not \sim c_2, $ then 
      $ \{ c_1, c_2, c_3 \} $ can be partitioned into
      $ \{ c_1 \}, \{ c_2, c_3 \}, $ so that Lemma~\ref{Lemma_disconnected} 
      can be applied.
      If we have $ c_1 \sim c_2, $ we can apply Lemma~\ref{Lemma_pivot}
      with $ C_1 = \{ c_1 \}, \ c = c_2, \ C_2 = \{ c_3 \}. $ 
   \item
      We prove $ S = 3. $ We use the fact
      that Conjecture~\ref{Conj_circles} holds for $ S < 3. $ 
      Let $ c_1, c_2, c_3, c_4 \in \Sigma^{+}. $ 
      If $ \{ c_1, c_2, c_3, c_4 \} $ can be partitioned into
      two disjoints sets, we can apply Lemma~\ref{Lemma_disconnected},
      and we are done. 
      Otherwise, if $ \{ c_1, c_2, c_3, c_4 \} $ cannot be
      partitioned into disconnected sets, there are two possibilities: 
      \begin{itemize}
      \item 
         The clauses form a line 
         $ c_1 \sim c_2 \sim c_3 \sim c_4. $ 
         If $ c_1 \sim c_4, $ then $ ( c_1, c_2, c_3, c_4 ) $ 
         is a circle, so that Conjecture~\ref{Conj_circles} holds trivially. 

         Otherwise, we can still have $ c_1 \sim c_3 $ or 
         $ c_2 \sim c_4. $ If we have both, then
         $ ( c_1, c_3, c_4, c_2 ) $ is a circle,
         so that Conjecture~\ref{Conj_circles} again holds trivially.

         If $ c_1 \not \sim c_3, $ we can apply
         Lemma~\ref{Lemma_pivot} with 
         $ C_1 = \{ c_1 \}, \ c = c_2, $ and $ C_2 = \{ c_3, c_4 \}. $ 
        
         Similarly, if $ c_2 \not \sim c_4, $ we can apply
         Lemma~\ref{Lemma_pivot} with 
         $ C_1 = \{ c_1, c_2 \}, \ c = c_3, $ and $ C_2 = \{ c_4 \}. $ 
      \item
         The clauses form a kind of star with $ c_1 $ in the center:
         $ c_1 \sim c_2, \ c_1 \sim c_3, \ c_1 \sim c_4. $ 

         For $ c_2, $ if $ c_2 \sim c_3, $ nor $ c_2 \sim c_4, $ 
         we can apply Lemma~\ref{Lemma_pivot} with 
         $ C_1 = \{ c_2 \}, \ c = c_1, \ C_2 = \{ c_3, c_4 \}. $ 

         If we have both of $ c_2 \sim c_3 $ and $ c_2 \sim c_4, $ then
         $ ( c_2,c_3,c_1,c_4 ) $ is a circle.  

         In the remaining case, we may assume without loss of generality
         that $ c_2 \sim c_3, $ but also $ c_2 \not \sim c_4. $ 

         This means that we have $ c_2 \sim c_3, c_2 \not \sim c_4. $ 
         If $ c_4 \sim c_3, $ then $ ( c_1,c_2,c_3,c_4) $ is a circle.
         If $ c_4 \not \sim c_3, $ then we can apply Lemma~\ref{Lemma_pivot} 
         with $ C_1 = \{ c_2,c_3 \}, \ c = c_1, \ C_2 = \{ c_4 \}. $ 
      \end{itemize}
   \end{itemize}
\end{proof}
We give a counter example for $ S = 4. $ 
\begin{exa}
   \label{Exa_failure_conj_circles}
   Consider the following GCSP, which has no blockings, 
   and the following clauses: 
   \[ \begin{array}{ll}
       (c_1) & (X_1,X_2,X_3) \ / \ (0,0,0) \ | \ (0,1,1) \ | \ 
                                   (1,1,0) \ | \ (1,0,1) \\
       (c_2) & (X_1,Y_1) \ / \ (0,0) \ | \ (1,1) \\
       (c_3) & (X_2,Y_2) \ / \ (0,0) \ | \ (1,1) \\
       (c_4) & (X_3,Y_3) \ / \ (0,0) \ | \ (1,1) \\
       (c_5) & (Y_1,Y_2,Y_3) \ / \ (1,0,0) \ | \ (0,1,0) \ | \ 
                                   (0,0,1) \ | \ (1,1,1) \\
      \end{array}
   \]
   We have $ c_1 \sim c_2, \ c_1 \sim c_3, \ c_1 \sim c_4, $ and
   $ c_2 \sim c_5, \ c_3 \sim c_5, \ c_4 \sim c_5. $ There are
   no other connections.
   The example can be understood as follows: Clause $ c_2 $ 
   requires that $ X_1 \Theta = Y_1 \Theta. $ 
   Similarly, $ c_3 $ requires that $ X_2 \Theta = Y_2 \Theta, $ 
   and $ c_4 $ requires that $ X_3 \Theta = Y_3 \Theta. $ 
   Clause $ c_1 $ requires that $ X_1 \Theta + X_2 \Theta + X_3 \Theta $
   is even, while $ c_5 $ requires that 
   $ Y_1 \Theta + Y_2 \Theta + Y_3 \Theta $ 
   is odd. Since the sums must be equal, and cannot be odd and even at 
   the same time, $ ( \{ c_1,c_2,c_3,c_4,c_5 \}, \{ \ \} ) $ has no 
   solution.
   
   Ignoring direction and starting point, there are three
   circles of size $ 4: $  
   \[ ( c_1,c_2,c_5,c_3 ), \ \
      ( c_1,c_2,c_5,c_4 ), \ \
      ( c_1,c_3,c_5,c_4 ). \]
   Since the circles are symmetric, we show that
   every substlet occurring in $ \{ c_1,c_2,c_5,c_3 \} $ can occur
   in a solution. One can pick the instance of $ c_1 $ and
   $ c_5 $ in such a way that they agree on $ X_1/Y_1 $ and
   $ X_2/Y_2. $ They will disagree on $ X_3/Z_3, $ but because
   $ c_4 $ is not considered, this is no problem.
   After that, the instances to $ c_2 $ and $ c_3 $
   are fixed. 
   It is easily checked that $ c_1,c_2,c_3,c_4,c_5 $ cannot be arranged
   into a circle. 
\end{exa}
Example~\ref{Exa_failure_conj_circles} contains a GCSP that would 
not be refined by Algorithm~\ref{Algo_local_consistency_checking} 
with $ S = 4, $ despite the fact that it has no solution. 
We will refrain from trying to make 
Algorithm~\ref{Algo_local_consistency_checking} complete, because 
we believe that it is not worth the effort. 
Experiments suggest that using 
Algorithm~\ref{Algo_local_consistency_checking} becomes
too costly already at $ S \geq 3. $ Implementing a more elaborate check
at $ S \geq 3, $ would 
make Algorithm~\ref{Algo_local_consistency_checking} even more costly, 
and harder to implement, without much hope for improvement.

It is important to observe that even when 
Algorithm~\ref{Algo_local_consistency_checking} is used
as a precheck, it still needs to be restorable, because 
it may be called by Algorithm~\ref{Alg_find_optimal_match},
which will turn on and off different substlets, based on $ \alpha. $
Only the first call need not be restorable. 

\section{Conclusions}

The problem of matching a geometric formula into an interpretation 
used to be the bottleneck of our implementation of geometric 
resolution. In order to 
improve this situation, we gave a translation of the matching problem into 
GCSP, and provided efficient approaches for solving GCSP. One 
approach is to solve the GCSP directly by a combination of refinement,
backtracking and learning. The other approach is to translate the problem
into SAT. Our experiments suggest that both
approaches have comparable performance. Both approaches still
have room for improvement. For translation to SAT, one could  
develop a dedicated SAT solver. For the direct approach using GCSP,
one could add heuristics that control when learnt lemmas are forgotten,
and probably it will be possible to add deep backtracking.

Independent of the relative performance of the two approaches, 
one can conclude that the speed of 
{\bf Geo} can be improved by a very large factor, and
that matching is no longer the bottleneck that hinders further
development of the geometric resolution calculus. 

Since GCSP may have applications outside of geometric logic,
we defined an input format for GCSP, similar to 
DIMACS format for SAT, that can be used for independent applications. 
We also made the sources
of our matching algorithm available. 

The fact that the clause refining algorithm based on
local consistency checking turned out not 
competitive, shows that search algorithms that appear good in theory, 
are not necessarily good in 
practice. In general, it is difficult to predict what will
be the effect of a modification of a search algorithm. 
A seemingly small change may have a large impact on performance.

As for geometric resolution, one might argue that a calculus that 
uses an NP-complete problem
as its basic operation is not viable, but there is room for interpretation: 
The complexity of the matching problem is caused by
the fact that as result of flattening terms, geometric formulas and 
interpretations 
have DAG-structure instead of tree-structure. This increased expressiveness 
means that a geometric formula possibly represents exponentially many
formulas with tree-structure. This may 
very well result in shorter proofs. Only experiments can determine 
which of the two effects will be stronger.

\section{Acknowledgements}

\noindent
We gratefully acknowledge that 
this work was supported by the Polish National Science Center
(Narodowe Centrum Nauki) under grant number DEC-2015/17/B/ST6/01898
(Application of Logic with Partial Functions). A large
part of this work was carried out while the author was employed
at Wroc{\l }aw University, Poland.

\bibliographystyle{plain}
\bibliography{subsumption}

\end{document}